\documentclass[onecolumn, a4size, 12pt]{IEEEtran}
\usepackage{amsmath}
\usepackage{graphicx}
\usepackage{epsfig}
\usepackage{latexsym}
\usepackage{amsfonts}
\usepackage{amssymb}
\usepackage{psfrag}
\usepackage{cite}
\usepackage{subfigure}
\usepackage{color}
\usepackage{float}
\usepackage{textcomp}
\usepackage{tabu}
\usepackage{footnote}
\usepackage{tabularx}
\usepackage{epstopdf}

\begin{document}
\pagestyle{plain}
\title{Capacity Region of MISO Broadcast Channel for Simultaneous Wireless Information and Power Transfer
\footnote{S. Luo, J. Xu and T. J. Lim are with the Department of
Electrical and Computer Engineering, National University of
Singapore (e-mail: \mbox{shixin.luo}@u.nus.sg; jiexu.ustc@gmail.com; eleltj@nus.edu.sg).} \footnote{R.
Zhang is with the Department of Electrical and Computer Engineering,
National University of Singapore (e-mail: elezhang@nus.edu.sg). He is
also with the Institute for Infocomm Research, A*STAR, Singapore.}\footnote{This paper has been presented in part at the IEEE International
Conference on Communications (ICC), London, 8-12 June, 2015.}}

\author{Shixin Luo, Jie Xu, Teng Joon Lim, and Rui
Zhang} 

\setlength{\textwidth}{7.1in} \setlength{\textheight}{9.7in}
\setlength{\topmargin}{-0.8in} \setlength{\oddsidemargin}{-0.30in}

\maketitle

\begin{abstract}
This paper studies a multiple-input single-output (MISO) broadcast channel (BC) featuring simultaneous wireless information and power transfer (SWIPT), where a multi-antenna access point (AP) delivers both information and energy via radio signals to multiple single-antenna receivers simultaneously, and each receiver implements either information decoding (ID) or energy harvesting (EH). In particular, pseudo-random sequences that are {\it a priori} known and therefore can be cancelled at each ID receiver are used as the energy signals, and the information-theoretically optimal dirty paper coding (DPC) is employed for the information transmission. We characterize the capacity region for ID receivers, by solving a sequence of weighted sum-rate (WSR) maximization (WSRMax) problems subject to a maximum sum-power constraint for the AP, and a set of minimum harvested power constraints for individual EH receivers. The problem corresponds to a new form of WSRMax problem in MISO-BC with combined maximum and minimum linear transmit covariance constraints (MaxLTCCs and MinLTCCs), which differs from the celebrated capacity region characterization problem for MISO-BC under a set of MaxLTCCs only and is challenging to solve. By extending the general BC-multiple access channel (MAC) duality, which is only applicable to WSRMax problems with MaxLTCCs, and applying the ellipsoid method, we propose an efficient iterative algorithm to solve this problem globally optimally. Furthermore, we also propose two suboptimal algorithms with lower complexity by assuming that the information and energy signals are designed separately. Finally, numerical results are provided to validate our proposed algorithms.
\end{abstract}


\IEEEpeerreviewmaketitle
\setlength{\baselineskip}{1.3\baselineskip}
\newtheorem{definition}{\underline{Definition}}[section]
\newtheorem{fact}{Fact}
\newtheorem{assumption}{Assumption}
\newtheorem{theorem}{\underline{Theorem}}[section]
\newtheorem{lemma}{\underline{Lemma}}[section]
\newtheorem{corollary}{Corollary}
\newtheorem{proposition}{\underline{Proposition}}[section]
\newtheorem{example}{\underline{Example}}[section]
\newtheorem{remark}{\underline{Remark}}[section]
\newtheorem{algorithm}{\underline{Algorithm}}[section]
\newcommand{\mv}[1]{\mbox{\boldmath{$ #1 $}}}

\section{Introduction}
Wireless energy transfer (WET) using radio frequency (RF) signals is a promising technology to provide perpetual power supplies for sensors, radio-frequency identification (RFID) tags, and other devices with very low power consumption and which are difficult to access \cite{Bimag,Ngmag}. In particular, RF-enabled WET enjoys many practical advantages, such as wide coverage, low production cost, small receiver form factor, and efficient energy multicasting thanks to the broadcast nature of electromagnetic waves. Of course, RF signals have also been widely used as a means for transmitting information. To enable a dual use of RF signals, simultaneous wireless information and power transfer (SWIPT) has become a fast-emerging area of research \cite{varshney08, sahai10, xun13,liang13,rui13,BF,xing13,shi14, Tao12,Hu14,Clerckx13,Otter14,Amin12,SH13, multicast1, multicast2, secure1, secure2}, where hybrid access points (APs) are deployed to simultaneously deliver both energy and information to one or more receivers via RF signals.

The idea of SWIPT was first proposed by Varshney \cite{varshney08}, in which a point-to-point single-antenna additive white Gaussian noise (AWGN) channel for SWIPT was investigated from an information-theoretic standpoint. This work was then extended to frequency-selective AWGN channels in \cite{sahai10}, where a non-trivial tradeoff between information rate and harvested energy was shown by varying power allocation over frequency. Prior works \cite{varshney08, sahai10} have studied the fundamental performance limits of wireless systems with SWIPT, where the receiver is ideally assumed to be able to decode the information and harvest the energy independently from the same received signal. However, this assumption implies that the received signal used for harvesting energy can be reused for decoding information without any loss, which is not realizable yet due to practical circuit limitations. Consequently, in \cite{xun13,rui13}, various practical receiver architectures for SWIPT were proposed, such as time-switching and power-splitting. The authors in \cite{liang13} studied SWIPT for fading AWGN channels subject to time-varying co-channel interference, and proposed a new principle termed ``opportunistic energy harvesting'' where the receiver switches between harvesting energy and decoding information based on the wireless channel condition and interference power level.

The practical implementation of SWIPT is limited by the severe path loss and fading of wireless channels, and multi-antenna processing is an appealing solution to improve the efficiency of both information and energy transfer. Recently, there have been a handful of papers on studying the multi-antenna SWIPT systems under various setups including broadcast channel (BC) \cite{rui13,xing13,BF,shi14,Tao12,secure1,secure2}, multicast system \cite{Hu14, multicast1, multicast2}, interference channel \cite{Otter14,Clerckx13,SH13}, and relay channel \cite{Amin12}. As for the multi-antenna BC, the authors in \cite{rui13} first characterized the rate-energy (R-E) tradeoff for a simplified multiple-input multiple-output (MIMO) BC with two (either separated or co-located) receivers implementing information decoding (ID) and energy harvesting (EH), respectively. The study in \cite{rui13} was then extended to the case with imperfect channel state information (CSI) at the transmitter \cite{Tao12}. Moreover, \cite{BF, xing13} and \cite{shi14} studied the multiple-input single-output (MISO) BC for SWIPT with multiple separated and co-located ID and EH receivers, respectively. In \cite{secure1, secure2}, physical layer security is considered under MISO BC for SWIPT by adding additional secrecy information transmission constraint, which reveal interesting new insights that the energy-carrying signal can also play the role of artificial noise (AN) to ensure secrecy in information transmission. However, all these prior works on multi-antenna BC consider low-complexity linear precoding/beamforming for SWIPT, which is in general suboptimal. Therefore, the fundamental limits on the information and energy transfer in general multi-antenna BC for SWIPT remain unknown, thus motivating this work.

This paper studies a MISO-BC for SWIPT, where a multi-antenna AP delivers both wireless information and energy to multiple receivers each with a single antenna. Each receiver implements either ID or EH alone\footnote{Conventional wireless information and energy receivers are respectively designed to operate with very different power requirements (e.g., an EH
receiver for a low-power sensor requires a received power of $-10$ dBm or more for real-time operation, while ID receivers such as cellular and WiFi mobile receivers often operate with a received power less than $-50$ dBm \cite{rui13}), and thus the existing RF front-end for wireless EH cannot currently be used for ID and vice versa.}. Pseudo-random sequences that are {\it a priori} known and therefore can be cancelled at each ID receiver are used as the energy signals, and the information-theoretically optimal dirty paper coding (DPC) \cite{dirty} is employed for the information transmission. Under this setup, we characterize the fundamental limits on the information and energy transfer of the considered MISO-BC for SWIPT, by establishing the capacity region for the ID receivers while ensuring given minimum energy requirements for EH receivers. Specifically, the capacity region is characterized by solving a sequence of weighted sum-rate (WSR) maximization (WSRMax) problems for all ID receivers subject to a maximum transmit sum-power constraint for the AP, and a set of minimum harvested power constraints for individual EH receivers. Interestingly, these problems belong to a new form of WSRMax problem for MISO-BC with combined maximum and minimum linear transmit covariance constraints (MaxLTCCs and MinLTCCs), which is non-convex in general and thus difficult to be solved optimally by standard convex optimization techniques.

It should be noted that the WSRMax problem with only MaxLTCCs has been investigated in \cite{UDD} to establish the capacity region of multi-antenna BC, in which a general BC-multiple access channel (MAC) duality is applied to solve this problem optimally. However, the WSRMax problem in our case is different and more challenging due to the newly introduced MinLTCCs that arise from the minimum harvested power constraints for the EH receivers. As a result, the general BC-MAC duality does not directly apply here. To overcome this challenge, we propose an efficient algorithm to optimally solve the new WSRMax problem with combined MaxLTCCs and MinLTCCs, by extending the general BC-MAC duality and applying the ellipsoid method. One more side effect of the MinLTCCs is that although the solution generated by the ellipsoid method achieves the optimal WSR, it may not be feasible to the primal problem. This is because the equivalent noise covariance matrix of the dual MAC may not be of full rank, which implies an infinite number of possible solutions for the dual MAC. In this case, a semi-definite program (SDP) needs to be further solved to obtain a primal feasible solution. To the best of our knowledge, our approach is novel and has not been studied in the literature. It is shown that at the optimal solution, the energy signals should be in the null space of all ID receivers' channels (if it is not an empty set). Furthermore, to reduce the implementation complexity of the optimal solution (especially the iterative search with the ellipsoid method), we propose two suboptimal algorithms by separately designing the information and energy signals. Finally, numerical results are provided to validate our proposed algorithms.

The remainder of this paper is organized as follows. Section \ref{sec:system model} introduces the system model and problem formulation. Section \ref{sec:sum power} and Section \ref{sec:sub} present the optimal and suboptimal solutions for the formulated problem, respectively. Section \ref{sec:numerical} provides numerical examples to validate our results. Finally, Section \ref{sec:conclusion} concludes this paper.

{\it Notations}: Boldface letters refer to vectors (lower case) or matrices (upper case). For a square matrix $\mv{S}$, $\text{Tr}(\mv{S})$ and $\mv{S}^{-1}$ denote its trace and inverse, respectively, while $\mv{S}\succeq \mv{0}$, $\mv{S}\preceq \mv{0}$ and $\mv{S} \nsucceq \mv{0}$ mean that $\mv{S}$ is positive semidefinite, negative semidefinite and non-positive semidefinite, respectively. For an arbitrary-size matrix $\mv{M}$, $\mv{M}^{H}$, $\mv{M}^{T}$ and $\mv{M}^\dag$ denote the conjugate transpose, transpose and pseudo-inverse of $\mv{M}$, respectively. The distribution of a circularly symmetric complex Gaussian (CSCG) random vector with mean vector $\mv{x}$ and covariance matrix $\boldsymbol\Sigma$ is denoted by $\mathcal{CN}(\mv{x},\boldsymbol\Sigma)$; and $\thicksim$ stands for ``distributed as''. $\mathbb{C}^{x\times y}$ denotes the space of $x\times y$ complex matrices. $\mathbb{E}[\cdot]$ denotes the statistical expectation. $\|\mv{x}\|$ denotes the Euclidean norm of a complex vector $\mv{x}$, and $|z|$ denotes the magnitude of a complex number $z$.

\section{System Model and Problem Formulation}\label{sec:system model}
\begin{figure}
\centering
\epsfxsize=0.7\linewidth
\includegraphics[width=11cm]{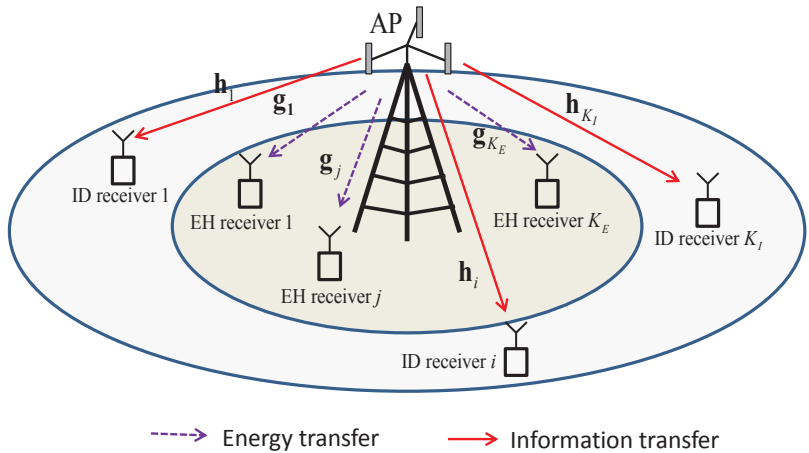}
\caption{A MISO broadcast system for simultaneous wireless information and
power transfer (SWIPT), where EH receivers are close to the AP for effective
energy reception.}
\label{fig:system}
\end{figure}
We consider a MISO-BC for SWIPT with an AP delivering both information and energy to multiple receivers over a single frequency band as shown in Fig. \ref{fig:system}, where each receiver implements either ID or EH. Note that our results apply to arbitrary user locations/channel realizations and there is no restriction on the locations of the EH/ID receivers. The receiver-location-based example in Fig. \ref{fig:system} is made for meeting the practically different received power requirements of EH and ID receivers. In this system, there are $K_I \geq 1$ ID receivers and $K_E \geq 1$ EH receivers, denoted by the sets $\mathcal{K_I}=\{1,\cdots,K_I\}$ and $\mathcal{K_E}=\{1,\cdots,K_E\}$, respectively. It is assumed that all ID and EH receivers are each equipped with one receive antenna, whereas the AP is equipped with $N > 1$ transmit antennas.

We assume a quasi-static channel model, and denote $\mv{h}_i \in \mathbb{C}^{N\times1}$ and $\mv{g}_j \in \mathbb{C}^{N\times1}$ as the channel vectors from the AP to ID receiver $i \in \mathcal{K_I}$ and to EH receiver $j \in \mathcal{K_E}$, respectively. The AP is assumed to perfectly know the instantaneous values of $\mv{h}_i$'s and $\mv{g}_j$'s, while each ID receiver knows its own instantaneous channel. In practice, the CSI of EH receivers can be acquired at the AP by e.g. reverse-link channel estimation based on training signals sent by the EH receivers via exploiting the channel reciprocity in time-division duplex (TDD) systems \cite{erchannel1}, or forward-link channel estimation and limited feedback by the EH receivers in frequency-division duplex (FDD) systems \cite{erchannel2, erchannel3}.

Without loss of generality, the AP transmits $K_I$ independent information signals, i.e., $\mv{x}_i \in \mathbb{C}^{N \times 1}, \forall i \in \mathcal{K_I}$, one for each ID receiver, and one common energy signal\footnote{Since the energy signal does not contain any information, one common energy signal with arbitrary rank covariance is sufficient to achieve the optimal energy transfer performance.}, i.e., $\mv{x}_E \in \mathbb{C}^{N \times 1}$, for all the EH receivers. Thus, the AP transmits the $N$-dimensional complex baseband signal
\begin{align}
\mv{x} = \sum_{i\in\mathcal{K_I}}\mv x_i+ \mv{x}_E
\end{align}
For information signals, we consider Gaussian signalling, and thus $\mv{x}_i$'s are independent and identically distributed (i.i.d.) CSCG vectors with zero mean and covariance matrix $\mv{S}_i \triangleq \mathbb{E}[\mv{x}_i\mv{x}^{H}_i]$, $i \in \mathcal{K_I}$. For the energy signal, since $\mv{x}_E$ does not carry any information, it can be implemented with a set of pseudo-random sequences that mimics a stationary $N$-dimensional random process with zero mean and covariance matrix $\mv{S}_E$.\footnote{Without loss of energy harvesting performance, we assume that the energy signal is pseudo-random instead of a deterministic sinusoidal wave, in order for its power spectral density to satisfy certain regulations on microwave radiation. Specifically, with pseudo-random energy signals, the transmit power spreads evenly over the operating frequency bands, which thus helps avoid a single power spike of the deterministic sinusoidal signal.} Suppose that the maximum sum-power at the AP is denoted by $P_{\text{sum}}>0$. Then we have $\mathbb{E}\left[\mv x^H \mv x\right] = \text{Tr}\left(\sum_{i\in \mathcal{K_I}}\mv{S}_i + \mv{S}_E\right) \le P_{\text{sum}}$.

We consider the information-theoretically optimal DPC for the information transmission, for which the causal interference can be pre-cancelled at the transmitter. To be more specific, consider the encoding order as $\pi(1),\ldots,\pi(K_I)$, i.e., the information signal $\mv x_{\pi(1)}$ for ID receiver $\pi(1)$ is encoded first, that for $\pi(2)$ is encoded second, and so on, where $\pi$ denotes some desired permutation over $\mathcal{K_I}$. In this case, for any ID receiver $\pi(i)$, the causal interference due to ID receivers $\pi(1),\ldots, \pi(i-1)$ can be canceled via DPC at the AP. As a result, the received signal for ID receiver $\pi(i)$ is expressed as
\begin{align}\label{eq:with eh int}
y_{\pi(i)} = \mv{h}_{\pi(i)}^H \mv{x}_{\pi(i)} + \sum_{k=i+1}^{K_I}\mv{h}_{\pi(i)}^H \mv{x}_{\pi(k)} + \mv{h}_{\pi(i)}^H\mv{x}_{E}+ z_i, i\in \mathcal{K_I}
\end{align}
where $z_i\thicksim\mathcal{CN}\left(0,\sigma^2\right)$ denotes the additive white Gaussian noise (AWGN) at the $i$th ID receiver with noise power being $\sigma^2$, and $\mv{h}_{\pi(i)}^H\mv{x}_{E}$ is the interference caused by the common energy signal.

Moreover, since the energy signal $\mv x_E$ is pseudo-random, its resulting interference can be efficiently cancelled by an extra interference cancellation operation at each ID receiver, explained as follows. Without loss of generality, $\mv{x}_E$ can be expressed as $\mv{x}_E = \sum^{L}_{l=1}\mv{v}_ls^{\text{E}}_l$ with $1 \leq L \leq N$ denoting the rank of $\mv{S}_E$, $\mv{v}_l$'s denoting the energy beamforming vectors each with unit norm, and $s^{\text{E}}_l$'s denoting the independently generated pseudo-random energy-bearing signals, whose waveforms can be assumed to be known at both the AP and each ID receiver. Given prior known $s^{\text{E}}_l$'s, in (\ref{eq:with eh int}) we have $\mv{h}_{\pi(i)}^H\mv{x}_{E} = \sum^{L}_{l=1}\left(\mv{h}_{\pi(i)}^H\mv{v}_l\right)s^{\text{EH}}_l$. By estimating the effective channel coefficients $\mv{h}_{\pi(i)}^H\mv{v}_l$'s at ID receiver $\pi(i)$, the resulting interference due to energy signals can be cancelled with known $s^{\text{E}}_l$'s. With the above interference cancellation, the received signal for ID receiver $\pi(i)$ in (\ref{eq:with eh int}) is re-expressed as
\begin{align}
y_{\pi(i)} = \mv{h}_{\pi(i)}^H \mv{x}_{\pi(i)} + \sum_{k=i+1}^{K_I}\mv{h}_{\pi(i)}^H \mv{x}_{\pi(k)} + z_i, i\in \mathcal{K_I}.
\end{align}

With Gaussian signalling employed, the achievable rate region for ID receivers, defined as the rate-tuples for all ID receivers (in bps/Hz) with given information covariance matrices $\{\mv{S}_i\}$, is thus given by \cite{BCCapacity}
\begin{align}
\mathcal{C}_{\text{BC}}&\left(\{\mv{S}_i\},\{\mv{h}_i\}\right) = \bigcup\limits_{\pi \in \boldsymbol\Pi}\left\{\mv{r} \in \mathbb{R}^{K_I}_+ : \right. \nonumber\\
&\left.r_{\pi(i)} \leq \log_2\left(\frac{\sigma^2+\mv{h}^{H}_{\pi(i)}\left(\sum^{K_I}_{k=i}\mv{S}_{\pi(k)}\right)\mv{h}_{\pi(i)}}{\sigma^2+\mv{h}^{H}_{\pi(i)}\left(\sum^{K_I}_{k=i+1}\mv{S}_{\pi(k)}\right)\mv{h}_{\pi(i)}}\right)\right\}
\end{align}
where $\boldsymbol\Pi$ is the collection of all possible permutations over $\mathcal{K_I}$, and $\mv{r} = [r_1,\ldots,r_{K_I}]^T$ denotes the vector of achievable rates for all ID receivers.

On the other hand, consider the WET. Due to the broadcast property of wireless channels, the energy carried by all information and energy signals can be harvested at each EH receiver. As a result, the harvested power for the $j$th EH receiver, denoted by $Q_j$, can be expressed as \cite{rui13}
\begin{align}
Q_j =\mathbb{E}\left[|\mv{g}_j^H \mv{x}|^2\right] = \zeta\text{Tr}\left[\left(\sum_{i\in\mathcal{K_I}} \mv{S}_i+\mv{S}_E\right)\mv{G}_j\right], j\in\mathcal{K_E}
\end{align}
where $0 < \zeta \leq 1$ denotes the energy harvesting efficiency \cite{xun13} at each EH receiver and $\mv{G}_j \triangleq \mv{g}_j\mv{g}_j^{H}, \forall j \in \mathcal{K_E}$.\footnote{Our results still hold when the antenna number of each EH receiver is larger than one, in which case the matrix $\mv{G}_j$'s for EH receivers are of higher rank instead of rank-one.} Since $\zeta$ is a constant, we normalize it as $\zeta = 1$ for simplicity unless otherwise specified.

Now, we are ready to present the optimization problem of interest. To characterize the boundary points of the capacity region for the MISO-BC with SWIPT, we maximize the WSR of all ID receivers subject to the minimum harvested power constraints at individual EH receivers, as well as the maximum sum-power constraint for the AP. By denoting the minimum harvested power requirement at EH receiver $j \in \mathcal{K_E}$ as $E_j > 0$, the WSRMax problem is formulated as
\begin{align}
\mathrm{(P1)}:~&\mathop{\mathtt{Max.}}\limits_{\left\{\mv{S}_i\right\},\mv{r},\mv{S}_E}
~~ \sum_{i\in \mathcal{K_I}}\alpha_ir_i  \label{eq:P1 objective}\\
\mathtt{s.t.}
& ~~ \mv{r} \in \mathcal{C}_{\text{BC}}\left(\{\mv{S}_i\},\{\mv{h}_i\}\right) \label{eq:P1 0}\\
& ~~ \text{Tr}\left[\left(\sum_{i\in\mathcal{K_I}} \mv{S}_i+\mv{S}_E\right)\mv{G}_j\right] \geq E_j, \forall j \in \mathcal{K_E} \label{eq:P1 1}\\
& ~~ \text{Tr}\left(\sum_{i\in\mathcal{K_I}} \mv{S}_i+\mv{S}_E\right) \leq P_{\text{sum}} \label{eq:P1 2}\\
& ~~ \mv{S}_E \succeq \mv{0}, \mv{S}_i \succeq \mv{0}, \forall i \in \mathcal{K_I} \label{eq:P1 3}
\end{align}
where $\alpha_i>0$ denotes a given weight for ID receiver $i \in \mathcal{K_I}$. Note that by solving problem (P1) via exhausting all possible $\{\alpha_i\}$, the whole capacity region can then be characterized. Let $\mathcal{D}$
denote the set containing all admissible information covariance matrices $\{\mv{S}_i\}$ and all achievable rates $\{r_i\}$ specified by the constraints in (\ref{eq:P1 0}) and (\ref{eq:P1 3}). It is then observed that (P1) is non-convex due to the non-convexity of $\mathcal{D}$, and thus the globally optimal solution of (P1) is difficult to obtain in general. Note that one commonly adopted approach to deal with this type of non-convex WSRMax problems for the multi-antenna BC is to use the BC-MAC duality to transform it into an equivalent convex WSRMax problem for a dual MAC \cite{UDD, UDD1, UDD2, UDD3}. However, the existing BC-MAC duality is only applicable to the case of MaxLTCCs\footnote{The MaxLTCC is expressed as $\text{Tr}\left(\mv{S}\mv{Q}\right) \leq P$, where $\mv{S}$ is the transmit covariance matrix to be optimized, $\mv{Q}$ is a given positive semi-definite matrix (which is identify matrix in (\ref{eq:P1 2})), and $P \geq 0$ is a prescribed power constraints. Note that our defined MaxLTCC is the same as the general LTCC (GLTCC) in \cite{UDD}.} with information signals. In contrast, (P1) has both a MaxLTCC in (\ref{eq:P1 2}) and a set of MinLTCCs\footnote{Similar to the MaxLTCC, the MinLTCC is defined as $\text{Tr}\left(\mv{S}\mv{Q}\right) \geq P$.} in (\ref{eq:P1 1}) as well as an energy covariance matrix $\mv{S}_E$. As a result, solving problem (P1) is not a trivial exercise, and has not been investigated yet in the literature. Note that our results are easily extendible to the case with per-antenna individual power constraints for the AP, for which the single MaxLTCC in (\ref{eq:P1 2}) is replaced by a set of MaxLTCCs as in \cite{UDD}. Also note that in this paper we focus on characterizing the fundamental limit of MISO-BC for SWIPT with given user channel realizations. The results can be extended to the general setup with time-varying (fading) channels where the channel capacity/harvested energy can be measured from either ergodic (average) or non-ergodic (outage) perspectives.

Prior to solving problem (P1), we first check its feasibility. It can be observed that (P1) is feasible if and only if its feasibility is guaranteed by ignoring all the ID receivers, i.e., setting $\mv{S}_i = \mv{0}$ and $r_i = 0,\forall i\in\mathcal{K_I}$. Thus, the feasibility of (P1) can be verified by solving the following problem:
\begin{align}
\mathop{\mathtt{find}} &
~~  \mv{S}_E \nonumber \\
\mathtt{s.t.}
& ~~ \text{Tr}\left[\mv{S}_E\mv{G}_j\right] \geq E_j, \forall j \in \mathcal{K_E} \nonumber \\
& ~~ \text{Tr}(\mv{S}_E) \leq P_{\text{sum}}, ~ \mv{S}_E \succeq \mv{0} \label{eq:feasibility}.
\end{align}
Since problem (\ref{eq:feasibility}) is a convex semi-definite program (SDP), it can be solved by standard convex optimization techniques such as the interior point method \cite{Boydbook}. In the rest of this paper, we only focus on the case that (P1) is feasible. In practice, (P1) can be infeasible due to e.g. poor channel conditions, insufficient transmit power or high minimum harvested power constraints. In such cases, the minimum harvested power constraints can be reduced (smaller $E_j$) for some EH receives to make (P1) feasible.

\section{Optimal Solution}\label{sec:sum power}
In this section, we present the optimal solution to problem (P1) by transforming it into a series of equivalent WSRMax sub-problems with a single MaxLTCC and accordingly solving these sub-problems via the BC-MAC duality. Specifically, we first define the following auxiliary function $g(\{\lambda_j\})$ as
\begin{align}
g(\{\lambda_j\}) = &\mathop{\mathtt{Max.}}\limits_{\left\{\mv{S}_i\right\},\mv{r},\mv{S}_E}
~~ \sum_{i \in \mathcal{K_I}}\alpha_ir_i  \label{eq:combined problem}\\
\mathtt{s.t.}
& ~~ \mv{r} \in \mathcal{C}_{\text{BC}}\left(\{\mv{S}_i\},\{\mv{h}_i\}\right) \\
& ~~ \sum_{i\in \mathcal{K_I}}\text{Tr}(\mv{A}\mv{S}_i) + \text{Tr}(\mv{A}\mv{S}_E) \leq P_{A} \label{eq:modified power} \\
& ~~ \mv{S}_E \succeq \mv{0}, \mv{S}_i \succeq \mv{0}, \forall i \in \mathcal{K_I}
\end{align}
where $\lambda_j \geq 0, j \in \{0\} \cup \mathcal{K_E}$ are auxiliary variables, $\mv{A} = \lambda_0\mv{I} - \sum_{j\in \mathcal{K_E}}\lambda_j\mv{G}_j$ and $P_{A} = \lambda_0P_{\text{sum}} - \sum_{j\in \mathcal{K_E}}\lambda_jE_j$. Note that $g(\{\lambda_j\})$ is generally not the dual function of problem (P1); however, it serves as an upper bound on the optimal value of (P1) for any $\{\lambda_j \ge 0\}$. This is because any feasible solution to problem (P1) is also feasible to (\ref{eq:combined problem}), but not necessarily vice versa. We then define the following problem by minimizing $g(\{\lambda_j\})$ over $\{\lambda_j\}$:
\begin{align}
\mathrm{(P2)}:~ \mathop{\mathtt{Min.}}\limits_{\{\lambda_j \geq 0\}} &
~~ g(\{\lambda_j\}).
\end{align}
In general the optimal value of problem (P2) also serves as an upper bound on that of (P1). However, as will be rigorously shown later (see Lemma \ref{lemma:3}), this upper bound is indeed tight. As a result, we will solve (P1) by equivalently solving problem (P2). In the following, we first solve problem (\ref{eq:combined problem}) to obtain $g(\{\lambda_j\})$ under any given $\{\lambda_j\geq0\}$, based on which the strong duality between problems (P1) and (P2) is then proved. Next, we solve problem (P2) to obtain the optimal $\{\lambda_j\}$, and finally, we construct the optimal solution to (P1) based on that to (P2).

\subsection{Solving Problem (\ref{eq:combined problem}) to Obtain $g(\{\lambda_j\})$ }\label{sec:BC MAC}
To start, we present some important properties of problem (\ref{eq:combined problem}) in the following lemma.
\begin{lemma}\label{lemma:1}
In order for problem (\ref{eq:combined problem}) to be feasible and $g(\{\lambda_j\})$ to have an upper-bounded value, i.e., $g(\{\lambda_j\}) < +\infty$, the following conditions must be satisfied:
\begin{enumerate}
\item $\mv{A}$ is positive semi-definite, i.e., $\mv{A} \succeq \mv{0}$.
\item The null space of $\mv{A}$ lies in the null space of $\mv{H} \triangleq \sum_{i\in\mathcal{K_I}} \mv{h}_i\mv{h}^{H}_{i}  \in \mathbb{C}^{N\times N}$, i.e., $\text{Null}(\mv{A}) \subseteq \text{Null}\left(\mv{H}\right)$, where $\text{Null}(\mv{A}) \triangleq \left\{\mv{x} \in \mathbb{C}^{N \times 1}: \mv{A}\mv{x} = \mathbf{0}\right\}$.
\item $P_{A} \geq 0$.
\end{enumerate}
\end{lemma}
\begin{proof}
See Appendix \ref{appendix:proof lemma 1}.
\end{proof}
From Lemma \ref{lemma:1}, it is sufficient for us to solve (\ref{eq:combined problem}) with $\mv{A}\succeq \mathbf{0}$, $\text{Null}(\mv{A}) \subseteq \text{Null}\left(\mv{H}\right)$ and $P_A\ge 0$.

Suppose that $\text{rank}(\mv{A}) = m$, where $\text{rank}(\mv{H}) \leq m \leq N$ due to the second condition in Lemma \ref{lemma:1}. Then, the singular value decomposition (SVD) of $\mv{A}$ can be expressed as
\begin{align}\label{eq:evd}
\mv{A} = \left[\mv{U}_1, \mv{U}_2\right]\boldsymbol\Lambda \left[\mv{U}_1, \mv{U}_2\right]^{H}
\end{align}
where $\mv{U}_1 \in \mathbb{C}^{N \times m}$ and $\mv{U}_2 \in \mathbb{C}^{N \times (N-m)}$ consist of the first $m$ and the last $N - m$ left singular vectors of $\mv{A}$, which correspond to the non-zero and zero singular values in $\boldsymbol\Lambda$, respectively. Therefore, the vectors in $\mv{U}_1$ and $\mv{U}_2$ form the orthogonal basis for the range and null space of $\mv{A}$, respectively. Then we have the optimal $\mv{S}_E$ for problem (\ref{eq:combined problem}), denoted by $\mv{\bar{S}}_E$, as follows.

\begin{lemma}\label{lemma:2}
The optimal energy covariance matrix in problem (\ref{eq:combined problem}) is expressed as
\begin{align}\label{eq:energy}
\mv{\bar{S}}_E = \mv{U}_2\mv{\bar{E}}\mv{U}_2^{H}
\end{align}
where $\mv{\bar{E}} \in \mathbb{C}^{(N-m)\times (N-m)}$ can be any positive semi-definite matrix. That is, any $\mv{\bar{S}}_E \succeq \mv{0}$ satisfying $\mv{A}\mv{\bar{S}}_E = \mv{0}$ is optimal to problem (\ref{eq:combined problem}). Note that when $m = N$, i.e., $\mv{A}$ is of full rank, $\mv{U}_2$ does not exist. In this case, we have $\mv{\bar{S}}_E = \mv{0}$.
\end{lemma}
\begin{proof}
See Appendix \ref{appendix:proof lemma 2}.
\end{proof}
Lemma \ref{lemma:2} shows that the optimal energy covariance matrix $\mv{\bar{S}}_E$ of problem (\ref{eq:combined problem}) lies in the null space of $\mv{A}$. By using this result, problem (\ref{eq:combined problem}) can thus be simplified to
\begin{align}
\mathop{\mathtt{Max.}}\limits_{\left\{\mv{S}_i\right\},\mv{r}} &
~~ \sum_{i \in \mathcal{K_I}}\alpha_ir_i  \nonumber \\
\mathtt{s.t.}
& ~~ \mv{r} \in \mathcal{C}_{\text{BC}}\left(\{\mv{S}_i\},\{\mv{h}_i\}\right) \nonumber \\
& ~~ \sum_{i\in \mathcal{K_I}}\text{Tr}(\mv{A}\mv{S}_i) \leq P_{A} \nonumber \\
& ~~ \mv{S}_i \succeq \mv{0}, \forall i \in \mathcal{K_I}. \label{eq:no energy}
\end{align}

Now, it remains to solve (\ref{eq:no energy}) to obtain the optimal information covariance matrices, denoted by $\{\mv{\bar{S}}_i\}$. Note that problem (\ref{eq:no energy}) corresponds to a WSRMax problem in MISO-BC under a single MaxLTCC. For the special case of $\mv{A}$ having full rank, this problem has been solved by the general BC-MAC duality \cite{UDD}. To handle the general case of $\mv{A}$ being rank deficient, which has not been addressed in the literature, we present the following lemma.

\begin{lemma}\label{lemma:4}
The optimal information covariance matrices, i.e., $\left\{\mv{\bar{S}}_i\right\}$, in problem (\ref{eq:no energy}) can be expressed as
\begin{align}\label{eq:optimal structure}
\mv{\bar{S}}_i = \mv{U}_1\mv{\bar{B}}_i\mv{U}^{H}_1 + \mv{U}_1\mv{\bar{C}}_i\mv{U}^{H}_2 &+ \mv{U}_2\mv{\bar{C}}^{H}_i\mv{U}^{H}_1 \nonumber \\
& + \mv{U}_2\mv{\bar{D}}_i\mv{U}^{H}_2, \forall i \in \mathcal{K_I}
\end{align}
where $\mv{\bar{B}}_i\in\mathbb{C}^{m\times m}$ is the unique solution of problem (\ref{eq:full rank}) below, $\mv{\bar{C}}_i \in \mathbb{C}^{m\times (N-m)}$ and $\mv{\bar{D}}_i \in \mathbb{C}^{(N-m)\times (N-m)}$ can be any matrices with appropriate dimensions such that $\mv{\bar{S}}_i \succeq \mv{0}$.
\begin{align}
\mathop{\mathtt{Max.}}\limits_{\left\{\mv{B}_i\right\},\mv{r}} &
~~ \sum_{i \in \mathcal{K_I}}\alpha_ir_i  \nonumber \\
\mathtt{s.t.}
& ~~ \mv{r} \in \mathcal{C}_{\text{BC}}\left(\{\mv{B}_i\},\{\mv{\hat{h}}_i\}\right) \nonumber \\
& ~~ \sum_{i\in \mathcal{K_I}}\text{Tr}(\mv{\hat{A}}\mv{B}_i) \leq P_{A} \nonumber \\
& ~~ \mv{B}_i \succeq \mv{0}, \forall i \in \mathcal{K_I} \label{eq:full rank}
\end{align}
where $\mv{\hat{h}}_i = \mv{U}^{H}_1\mv{h}_i \in \mathbb{C}^{m\times 1}, \forall i \in \mathcal{K_I}$ and $\mv{\hat{A}} = \mv{U}^{H}_1\mv{A}\mv{U}_1 \in \mathbb{C}^{m\times m}$.
\end{lemma}
\begin{proof}
See Appendix \ref{appendix:proof lemma 4}.
\end{proof}

Note that $\mv{\hat{A}}$ is of full rank, and thus problem (\ref{eq:full rank}) can be solved by the general BC-MAC duality as in \cite{UDD}. By combining Lemmas \ref{lemma:2} and \ref{lemma:4}, we obtain the optimal solution to (\ref{eq:combined problem}).

\begin{remark}\label{remark:1}
Note that if $\mv{A}$ is of full rank, i.e., $m = N$, then $\mv{U}_2$ does not exist. In this case, the optimal solution to (11) is unique and can be expressed as
\begin{align}
\mv{\bar{S}}_i = \mv{U}_1\mv{\bar{B}}_i\mv{U}^{H}_1, \forall i \in \mathcal{K_I} ~ \text{and} ~\mv{\bar{S}}_E = \mv{0}. \label{eq:solution for dual}
\end{align}
However, if $\mv{A}$ is rank deficient, i.e., $m < N$, then $\mv{U}_2$ does exist in general. In this case, there exist infinite sets of optimal solution $\left\{\{\mv{\bar{S}}_i\},\mv{\bar{S}}_E\right\}$ based on Lemmas \ref{lemma:2} and \ref{lemma:4}, and as a result the optimal solution to problem (\ref{eq:combined problem}) is not unique. For simplicity, we employ the specific optimal solution in (\ref{eq:solution for dual}) to solve (\ref{eq:combined problem}) for obtaining $g(\{\lambda_j\})$.
\end{remark}

\subsection{Solving Problem (P2)}\label{sec:2}
In this section, we first prove the strong duality between (P1) and (P2) before solving (P2) to find the optimal $\{\lambda_j\}$ for maximizing $g(\{\lambda_j\})$.
\begin{lemma}\label{lemma:3}
The optimal value of problem (P1) is equal to that of problem (P2).
\end{lemma}
\begin{proof}
See Appendix \ref{appendix:proof lemma 3}.
\end{proof}

Next, we proceed to solve (P2). Since $g(\{\lambda_j\})$ is upper bounded only when the conditions in Lemma \ref{lemma:1} are satisfied, we can rewrite (P2) as follows by adding these conditions as explicit constraints.
\begin{align}
\mathrm{(P3)}:~ \mathop{\mathtt{Min.}}\limits_{\{\lambda_j \geq 0\}} &
~~ g(\{\lambda_j\}) \label{eq:p22}\\
\mathtt{s.t.}
& ~~ \text{Null}(\mv{A}) \subseteq \text{Null}\left(\mv{H}\right) \label{eq:p3 1}\\
& ~~  \lambda_0\mv{I} - \sum^{K_E}_{j=1}\lambda_j\mv{G}_j \succeq \mv{0} \label{eq:p3 2}\\
& ~~ \lambda_0P_{\text{sum}} - \sum^{K_E}_{j=1}\lambda_jE_j \geq 0 \label{eq:p3 3}.
\end{align}
Note that for problem (P3), the objective function $g(\{\lambda_j\})$ is not necessarily differentiable. Nonetheless, we have the following lemma.
\begin{lemma}\label{lemma:6}
For the function $g(\{\lambda_j\})$ at any two non-negative points $[\dot{\lambda}_0,\dot{\lambda}_1,\cdots,\dot{\lambda}_{K_E}]$ and $[\ddot{\lambda}_0,\ddot{\lambda}_1,\cdots,\ddot{\lambda}_{K_E}]$, we have
\begin{align}
&g(\{\dot{\lambda}_j\}) \geq g(\{\ddot{\lambda}_j\}) + c\left[P_{\text{sum}} - \text{Tr}(\mv{\ddot{S}}_I),\text{Tr}(\mv{\ddot{S}}_I\mv{G}_1)-E_1,\cdots, \right.\nonumber \\
&\left.\text{Tr}(\mv{\ddot{S}}_I\mv{G}_{K_E})-E_{K_E}\right] \left([\dot{\lambda}_0,\dot{\lambda}_1,\cdots,\dot{\lambda}_{K_E}] - [\ddot{\lambda}_0,\ddot{\lambda}_1,\cdots,\ddot{\lambda}_{K_E}]\right)^{T}
\end{align}
where $\mv{\ddot{S}}_I = \sum_{i\in \mathcal{K_I}} \mv{\ddot{S}}_i$ with $\{\mv{\ddot{S}}_i\}$ being the optimal solution of problem (\ref{eq:combined problem}) given $\lambda_j = \ddot{\lambda}_j, j = 0,1,\cdots,K_E$, and $c \geq 0$ is a constant.
\end{lemma}
\begin{proof}
The proof is similar to that of \cite[Proposition 6]{UDD}, and is thus omitted for brevity.
\end{proof}
Lemma \ref{lemma:6} ensures that compared to the arbitrary point $[\ddot{\lambda}_0,\ddot{\lambda}_1,\cdots,\ddot{\lambda}_{K_E}]$, the optimal point that minimizes $g(\{\lambda_j\})$ cannot belong to the set of points $[\dot{\lambda}_0,\dot{\lambda}_1,\cdots,\dot{\lambda}_{K_E}]$ with
\begin{align}
&\left[P_{\text{sum}} - \text{Tr}(\mv{\ddot{S}}_I),\text{Tr}(\mv{\ddot{S}}_I\mv{G}_1)-E_1,\cdots,\text{Tr}(\mv{\ddot{S}}_I\mv{G}_{K_E})-E_{K_E}\right] \nonumber \\
&\cdot\left([\dot{\lambda}_0,\dot{\lambda}_1,\cdots,\dot{\lambda}_{K_E}] - [\ddot{\lambda}_0,\ddot{\lambda}_1,\cdots,\ddot{\lambda}_{K_E}]\right)^{T} > 0
\end{align}
and thus this set should be eliminated when searching for the optimal $\{\lambda_j\}$. This property motivates us to use the ellipsoid method \cite{Boyd2} to solve problem (P3). In order to successfully implement the ellipsoid method, we need to further obtain the sub-gradients for the constraints $\text{Null}(\mv{A}) \subseteq \text{Null}\left(\mv{H}\right)$ in (\ref{eq:p3 1}) and $\lambda_0\mv{I} - \sum^{K_E}_{j=1}\lambda_j\mv{G}_j \succeq \mv{0}$ in (\ref{eq:p3 2}), which are shown in the following two lemmas.

\begin{lemma}\label{lemma:5}
The constraint in (\ref{eq:p3 1}) is equivalent to the following linear constraints
\begin{align}
f_l(\{\lambda_j\}) \triangleq - \lambda_0 + \sum^{K_E}_{j=1}\lambda_j|\mv{v}^{H}_l\mv{g}_j|^2 < 0, \forall l \leq t \label{eq:P4 1}
\end{align}
where $t$ denotes the rank of matrix $\mv{H}$, and $\mv{v}_l, l=1,\cdots,t$, denote the $t$ left singular vectors of $\mv{H}$ corresponding to its non-zero singular values. As a result, the sub-gradient of $f_l(\{\lambda_j\})$ at given $\{\lambda_j\}$ can be expressed as $\left[-1, |\mv{v}^{H}_l\mv{g}_1|^2,\cdots,|\mv{v}^{H}_l\mv{g}_{K_E}|^2\right]^T$, $l=1,\cdots,t$.
\end{lemma}
\begin{proof}
See Appendix \ref{appendix:proof lemma 7}.
\end{proof}
\begin{lemma}\label{lemma:7}
Define $\mv{F}(\{\lambda_j\})= - \lambda_0\mv{I} + \sum^{K_E}_{j=1}\lambda_j\mv{G}_j$. Then the constraint in (\ref{eq:p3 2}) is equivalent to $\mv{F}(\{\lambda_j\}) \preceq \mv{0}$. Let $\mv{z}$ denote the dominant eigenvector of $\mv{F}(\{\lambda_j\})$, i.e., $\mv{z} = \arg\max\limits_{\|\mv{z}\|=1}\mv{z}^{H}\mv{F}(\{\lambda_j\})\mv{z}$. Then, the sub-gradient of $\mv{F}(\{\lambda_j\})$ at given $\{\lambda_j\}$ is $\left[-\|\mv{z}\|^2,\mv{z}^{H}\mv{G}_1\mv{z},\cdots,\mv{z}^{H}\mv{G}_{K_E}\mv{z}\right]^T$.
\end{lemma}
\begin{proof}
Also see Appendix \ref{appendix:proof lemma 7}.
\end{proof}

With Lemma \ref{lemma:6} and the sub-gradients in Lemmas \ref{lemma:5} and \ref{lemma:7} in hand, we can successfully solve problem (P2) by applying the ellipsoid method to update $\{\lambda_j\}$ towards the optimal solution $\{\lambda^{*}_{j}\}$.

\begin{remark}\label{remark:2}
Although we cannot prove the convexity of problem (P2), the convergence of the ellipsoid method can be ensured as explained in the following. The Lagrangian function of problem (P1) can be written as
\begin{align}\label{eq:re 1}
\sum^{K_I}_{i=1}\alpha_ir_i + \sum^{K_E}_{j=1}&\left[\theta_j\left(\text{Tr}\left[\left(\sum_{i \in \mathcal{K_I}}\mv{S}_i+\mv{S}_E\right)\mv{G}_j\right]-E_j\right)\right. \nonumber \\ &\left.- \theta_0\left(\text{Tr}\left[\sum_{i \in \mathcal{K_I}}\mv{S}_i+\mv{S}_E\right]-P_{\text{sum}}\right)\right]
\end{align}
where $\theta_0$ and $\{\theta_j\}^{K_E}_{j=1}$ are the Lagrange multipliers with respect to the constraints in (\ref{eq:P1 2}) and (\ref{eq:P1 1}), respectively. On the other hand, the Lagrangian function of problem (\ref{eq:combined problem}) can be written as
\begin{align}\label{eq:re 2}
&\sum^{K_I}_{i=1}\alpha_ir_i -\beta \left[\lambda_0\text{Tr}\left[\sum_{i \in \mathcal{K_I}}\mv{S}_I+\mv{S}_E\right] \right. \nonumber \\ &\left.- \sum^{K_E}_{j=1}\lambda_j\text{Tr}\left[\left(\sum_{i \in \mathcal{K_I}}\mv{S}_I+\mv{S}_E\right)\mv{G}_j\right] -\lambda_0P_{\text{sum}} + \sum^{K_E}_{j=1}\lambda_jE_j\right]
\end{align}
where $\beta$ is the Lagrange variable associated with the constraint of (\ref{eq:modified power}). By observing (\ref{eq:re 1}) and (\ref{eq:re 2}), we can see that the two Lagrangian functions are identical to each other if we choose $\theta_j = \beta\lambda_j, \forall j \in \mathcal{K_E}$. Thus, the auxiliary variables $\{\lambda_j\}$ can be viewed as the scaled (by a factor of $1/\beta$) Lagrange dual variables of problem (P1). Correspondingly, $g(\{\lambda_j\})$ is related to the dual function of problem (P1), which is known to be convex. However, since the optimal dual solution for $\beta$ in problem (\ref{eq:combined problem}) varies with $\{\lambda_j\}$, $g(\{\lambda_j\})$, it is not necessarily a convex function. Nevertheless, the above relationship reveals that in Lemma \ref{lemma:6}, the vector $\left[P_{\text{sum}} - \text{Tr}(\mv{\ddot{S}}_I),\text{Tr}(\mv{\ddot{S}}_I\mv{G}_1)-E_1,\cdots,\text{Tr}(\mv{\ddot{S}}_I\mv{G}_{K_E})-E_{K_E}\right]$ is indeed the exact sub-gradient for the convex dual function of problem (P1), given the fact that $\mv{\ddot{S}}_E$ with $\text{Tr}(\mv{A}\mv{\ddot{S}}_E) = 0$ is optimal to problem (\ref{eq:combined problem}) from Lemma \ref{lemma:2}. Thus, the convergence of the ellipsoid method based on this sub-gradient is guaranteed.
\end{remark}

\subsection{Finding Primal Optimal Solution to (P1)}\label{sec:3}
So far, we have obtained the optimal solution to (P2), i.e., $\{\lambda_j^*\}$, as well as the corresponding optimal solution to (\ref{eq:combined problem}) given in (\ref{eq:solution for dual}). According to Remark \ref{remark:1}, if $\mv{A}^{*}\triangleq \lambda^{*}_0\mv{I} - \sum^{K_E}_{j=1}\lambda^{*}_j\mv{G}_j$ is of full rank, (\ref{eq:solution for dual}) is the unique solution to (\ref{eq:combined problem}), which is thus optimal to (P1). However, if $\mv{A}^{*}$ is not of full rank, (\ref{eq:solution for dual}) is not the unique solution to (\ref{eq:combined problem}), and thus may not meet the minimum harvested power constraints in (\ref{eq:P1 1}). In the latter case, we need to find one feasible (thus optimal) solution of (P1), denoted by $\left\{\{\mv{S}^{*}_i\},\mv{S}^{*}_E\right\}$, from all the optimal solutions of (\ref{eq:combined problem}) given in (\ref{eq:energy}) and (\ref{eq:optimal structure}) with $\{\lambda_j^*\}$.

Denote the SVD of $\mv{A}^*$ as $\left[\mv{U}^{*}_1, \mv{U}^{*}_2\right]\boldsymbol\Lambda^{*} \left[\mv{U}^{*}_1, \mv{U}^{*}_2\right]^{H}$. Then following (\ref{eq:energy}) and (\ref{eq:optimal structure}), we can write the information and energy covariance matrices as
\begin{align}
\mv{S}_i = \mv{U}^{*}_1\mv{B}^{*}_i\left(\mv{U}^{*}_1\right)^{H} &+ \mv{U}^{*}_1\mv{C}_i\left(\mv{U}^{*}_2\right)^{H} + \mv{U}^{*}_2\mv{C}_i^{H}\left(\mv{U}^{*}_1\right)^{H} \nonumber \\
&+ \mv{U}^{*}_2\mv{D}_i\left(\mv{U}^{*}_2\right)^{H} , \forall i \in \mathcal{K_I} \label{eq:information}\\
\mv{S}_E = \mv{U}^{*}_2\mv{E}\left(\mv{U}^{*}_2\right)^{H}\label{eq:energy final}
\end{align}
where $\mv{B}^{*}_i$ is obtained by solving (\ref{eq:full rank}) with $\{\lambda_j^*\}$. Therefore, it remains to find a feasible and optimal set of $\{\mv{C}_i\}$, $\{\mv{D}_i\}$ and $\mv{E}$ such that the minimum harvested power constraints in (P1) are all satisfied. Since $r^{*}_i$ does not depend on the choice of $\{\mv{C}_i\}$, $\{\mv{D}_i\}$ and $\mv{E}$, $\forall i \in \mathcal{K_I}$, finding the primal optimal solution corresponds to solving a feasibility problem only involving the constraints in (\ref{eq:P1 2}) and (\ref{eq:P1 3}). Note that in general, there can be more than one feasible solutions to such a feasibility problem. Among them, we are interested in the solution with low-rank information covariance matrices in order to minimize the decoding complexity at the ID receiver. Therefore, we propose to minimize the sum of the ranks of all information covariance matrices, i.e., $\sum_{i\in \mathcal{K_I}} \text{rank}(\mv{S}_i)$. However, the rank function is not convex. By applying the convex approximation of the rank function \cite{rank} and using the fact that the nuclear norm of a covariance matrix equals to its trace, we solve the following problem to find a desired optimal solution.
\begin{align}
\mathrm{(P4)}:&\mathop{\mathtt{Min.}}\limits_{\left\{\mv{C}_i\right\}, \left\{\mv{D}_i\right\}, \mv{E}}
~~ \sum_{i \in \mathcal{K_I}}\text{Tr}\left(\mv{S}_i\right) \\
\mathtt{s.t.}
& ~~ \text{Tr}\left[\left(\sum_{i \in \mathcal{K_I}}\mv{S}_i+\mv{S}_E\right)\mv{G}_j\right] \geq E_j, \forall j \in \mathcal{K_E} \label{eq:p4 1}\\
& ~~ \text{Tr}\left(\sum_{i \in \mathcal{K_I}}\mv{S}_i+\mv{S}_E\right) \leq P_{\text{sum}} \label{eq:p4 2}\\
& ~~ \mv{S}_E \succeq \mv{0}, \mv{S}_i \succeq \mv{0}, \forall i \in \mathcal{K_I} \label{eq:p4 3}
\end{align}
where $\mv{S}_i$ and $\mv{S}_E$ are given in (\ref{eq:information}) and (\ref{eq:energy final}), respectively. As a result, the primal optimal solution to (P1) is finally obtained. Note that in the case that the obtained solution of (P2) with (\ref{eq:solution for dual}) is not feasible to (P1), the information covariance matrices may need to be expanded according to (\ref{eq:information}) if adding dedicated energy signal still cannot satisfy the energy requirements of all EH receivers. By combining the procedures in Sections \ref{sec:BC MAC}, \ref{sec:2} and \ref{sec:3},  the overall algorithm for solving problem (P1) is summarized in Table \ref{table1}.

For the algorithm given in Table \ref{table1}, the computation time is dominated by the ellipsoid method in steps 1)-3) and the SDP in step 4). In particular, the time complexity of step a) is of order $m^3K^2_I + m^2K^3_I$ by standard interior point method \cite{Boyd2}, where $m$ is the rank of matrix $\mv{A}$. Therefore, the worst case complexity is of order $N^3K^2_I + N^2K^3_I$. For step b), the complexity for computing the sub-gradients of $g(\{\lambda_j\})$ and the constraints in (\ref{eq:p3 1}), (\ref{eq:p3 2}) and (\ref{eq:p3 3}) is of order $N^2K_E$, and that for updating $\{\lambda_j\}$ is of order $K_E^2$. Thus, the time complexity of steps a)-b) is $\mathcal{O}(N^3K^2_I+N^2K^3_I+N^2K_E+K^2_E)$ in total. Note that step
2) iterates $\mathcal{O}(K^2_E)$ times to converge \cite{Boyd2}, thus the total time complexity of steps 1)-3) is $\mathcal{O}\left(K^2_E(N^3K^2_I+N^2K^3_I+N^2K_E+K^2_E)\right)$. The time complexity of solving SDP in step 4) is $\mathcal{O}(K^3_IN^{3.5}+K^4_I)$ \cite{complexity}. Thus, the overall complexity is $\mathcal{O}\left(K^2_E(N^3K^2_I+N^2K^3_I+N^2K_E+K^2_E) + K^3_IN^{3.5}+K^4_I\right)$ at most for the algorithm in Table \ref{table1}.

\begin{table}[ht]
\begin{center}
\caption{\textbf{Algorithm 1}: Algorithm for Solving Problem (P1)} \vspace{0.2cm}
 \hrule
\vspace{0.3cm}
\begin{enumerate}
\item Initialize $\lambda_j \geq 0, \forall j \in \mathcal{K_E}$.
\item {\bf Repeat:}
    \begin{itemize}
    \item[ a)] Obtain $\{\mv{\bar{S}}_i\}$ by solving problem (\ref{eq:full rank}) with given $\{\lambda_j\}$;
    \item[ b)] Compute the sub-gradients of $g(\{\lambda_j\})$ and the constraints in (\ref{eq:p3 1}), (\ref{eq:p3 2}) and (\ref{eq:p3 3}), and update $\{\lambda_j\}$ accordingly using the ellipsoid method \cite{Boydbook}.
    \end{itemize}
\item {\bf Until }$\{\lambda_j\}$ converges within a prescribed accuracy.
\item Set $\lambda_j^* = \lambda_j, \forall j\in \mathcal{K_E}$. If $\mv{A}^{*}$ is not of full rank and the obtained solution by (\ref{eq:solution for dual}) is not feasible to (P1), then find the optimal covariance matrices for information and energy transfer by solving problem (P4).
\end{enumerate}
\vspace{0.2cm} \hrule \label{table1} \end{center}
\end{table}

It is worth pointing out that in general there exist three cases for the optimal solution of (P1) obtained by the algorithm in Table \ref{table1}. For convenience, we denote $\mv{S}^{*}_I = \sum_{i \in \mathcal{K_I}}\mv{S}^{*}_i$.
\begin{enumerate}
\item $\mv{S}^{*}_I = \mv{0}$ and $\mv{S}^{*}_E \succeq \mv{0}$: in this case, no information can be transferred without violating the minimum harvested power constraints. This situation only occurs when the channel of each ID receiver is orthogonal to that of any EH receiver (i.e., $\mv{h}_i^H\mv{g}_j = 0, \forall i\in \mathcal{K_I}, j\in \mathcal{K_E}$), and full transmit power is used for ensuring the harvested power constraints. Note that under practical setup with randomly generated wireless channels, this case does not occur.
\item $\mv{S}^{*}_I \succeq \mv{0}$ and $\mv{S}^{*}_E = \mv{0}$: in this case, no dedicated energy signal is required. This corresponds to that the energy harvested from the information signals at each EH receiver is sufficient to satisfy the harvested power constraints. One situation for this case to occur is that if $\mv{H}$ is of full rank, then $\mv{A}^{*}$ is also full rank such that the unique optimal solution to problem (\ref{eq:combined problem}) (and thus optimal to (P1)) is $\mv{S}^{*}_i = \mv{U}^{*}_1\mv{B}^{*}_i\left(\mv{U}^{*}_1\right)^{H}, \forall i \in \mathcal{K_I}$, and $\mv{S}^{*}_E = \mathbf{0}$ from Remark \ref{remark:1}.
\item $\mv{S}^{*}_I \succeq \mv{0}$ and $\mv{S}^{*}_E \succeq \mv{0}$: in this case, dedicated energy signal is required to guarantee the harvested power constraints while maximizing the WSR. Interestingly, given the strong duality between (P1) and (P2) as well as Lemma \ref{lemma:2}, the optimal dedicated energy signal is orthogonal to the MISO channels of all the ID receivers. \textbf{Therefore, the extra processing of pre-canceling the interference caused by energy signals at the AP (via DPC) or at each ID receiver is not needed}.
\end{enumerate}
Note that the obtained optimal information and energy covariance matrices can have a rank larger than unity in general. However, our extensive simulation trials show that Algorithm 1 always returns rank-one information covariance matrices thanks to the approximated rank minimization employed in (P4). Nevertheless, it is difficult for us to guarantee the existence of optimal rank-one information covariance matrices in general.

\begin{remark}\label{remark:3}
To further provide insights on the transmit covariance matrices expansion in (P4), we present an intuitive explanation on how the obtained $\{\mv{C}_i^*\}$ and $\{\mv{D}_i^*\}$ can help ensure the harvested power constraints at all the EH receivers. First, denote the middle two terms in (\ref{eq:information}) involving $\{\mv{C}_i\}$ as $\mv{O}_i$, i.e., $\mv{O}_i = \mv{U}^{*}_1\mv{C}_i\left(\mv{U}^{*}_2\right)^{H} + \mv{U}^{*}_2\mv{C}_i^{H}\left(\mv{U}^{*}_1\right)^{H}, \forall i \in \mathcal{K_I}$. Since the columns of $\mv{U}^{*}_1$ and $\mv{U}^{*}_2$ form the orthogonal basis for the range and null space of $\mv{A}^{*}$, $\mv{O}_i$ has the following two properties:
\begin{align}
\text{Tr}\left(\mv{O}_i\right) & = 0, \forall i \in \mathcal{K_I} \label{eq:pro1}\\
\text{Tr}\left(\mv{O}_i\mv{A}^{*}\right) & = 0, \forall i \in \mathcal{K_I}. \label{eq:pro2}
\end{align}
Based on (\ref{eq:pro1}) and (\ref{eq:pro2}) with some manipulations, it follows that
\begin{align}
\sum^{K_E}_{j=1}\lambda^{*}_j\text{Tr}\left(\mv{O}_i\mv{G}_j\right) = 0, \forall i \in \mathcal{K_I}. \label{eq:pro3}
\end{align}
From (\ref{eq:pro1}), it is observed that $\mv{O}_i$'s do not cost any transmission power. Furthermore, according to (\ref{eq:pro2}) and the second point of Lemma \ref{lemma:1}, $\mv{O}_i$'s do not affect the data rate of ID receivers, i.e., $r_i, \forall i \in \mathcal{K_I}$. However, based on (\ref{eq:pro3}), it is observed that $\mv{O}_i$'s serve the purpose of re-allocating the power harvested at each EH receiver from the information embedded signals without affecting the data rate of each ID receiver. This reallocation is necessary if there exists $j \in \mathcal{K_E}$ such that $\text{Tr}\left[\left(\sum_{i \in \mathcal{K_I}}\mv{U}^{*}_1\mv{B}^{*}_i\left(\mv{U}^{*}_1\right)^{H}+\mv{S}^{*}_E\right)\mv{G}_j\right] < E_j$. Finally, from the theory of Schur complement \cite{Boydbook}, it is known that $\mv{S}^{*}_i \succeq \mv{0}$ if and only if (iff) the following conditions are satisfied:
\begin{align}
\mv{B}^{*}_i &\succeq \mv{0} \\
\left(\mv{I} - \mv{B}_i^{*}(\mv{B}_i^{*})^{\dag}\right)\mv{C}_i &= \mv{0} \\
\mv{D}_i - \mv{C}^{H}_i(\mv{B}_i^{*})^{\dag}\mv{C}_i &\succeq \mv{0}.
\end{align}
Therefore, $\mv{D}_i$ may be required to ensure that $\mv{S}^{*}_i \succeq \mv{0}$.
\end{remark}

\section{Separate Information and Energy Signal Design}\label{sec:sub}
So far, we have optimally solved problem (P1) by jointly designing the information and energy signals, which however requires significant computational complexity due to the iterative implementation based on the ellipsoid method. To reduce the complexity, in this section, we propose two suboptimal algorithms with separate information and energy signal design, namely ID/EH oriented separate information and energy signal design (IDSIED/EHSIED). Note that since separate information and energy signal design is assumed in both suboptimal algorithms, unlike the optimal solution, the energy signal is in general not orthogonal to the channels of all the ID receivers; as a result, the extra processing of pre-canceling the interference caused by energy signals at the AP (via DPC) or at each ID receiver is needed.

\subsection{ID Oriented Separate Information and Energy Signal Design}\label{sec:sub1}
In this algorithm, the total transmit power $P_{\text{sum}}$ is divided into two parts: $P_I$ and $P_{\text{sum}} - P_I$ ($0\leq P_I\leq P_\text{sum}$), which are exclusively allocated to information and energy signals, respectively. With any given power allocation, i.e., $P_I$, the information covariance matrices $\{\mv{S}_i\}$ are first designed to maximize the WSR for all ID receivers, by solving the following optimization problem:
\begin{align}
\mathop{\mathtt{Max.}}\limits_{\left\{\mv{S}_i\right\},\mv{r}} &
~~ \sum^{K_I}_{i=1}\alpha_ir_i  \nonumber \\
\mathtt{s.t.}
& ~~ \{r_i\} \in \mathcal{C}_{\text{BC}}\left(\{\mv{h}_i\},\{\mv{S}_i\}\right) \nonumber \\
& ~~ \text{Tr}\left(\sum_{i \in \mathcal{K_I}}\mv{S}_i\right) \leq P_{I} \nonumber \\
& ~~ \mv{S}_i \succeq \mv{0}, \forall i \in \mathcal{K_I}. \label{eq:ps information}
\end{align}
Note that problem (\ref{eq:ps information}) is a WSRMax problem in MISO-BC under a single MaxLTCC, which can be solved by the general BC-MAC duality with a guaranteed rank-one solution.

Next, let $\left\{\mv{S}^{'}_i(P_{I})\right\}$ be the optimal solution of (\ref{eq:ps information}) given $P_{I} \geq 0$. The energy covariance matrix $\mv{S}_E$ is then optimized to ensure that the harvested power constraint of each EH receiver is satisfied, as follows:
\begin{align}
\mathop{\mathtt{find}} &
~~ \mv{S}_E  \nonumber \\
\mathtt{s.t.}
& ~~ \text{Tr}\left[\left(\sum_{i \in \mathcal{K_I}}\mv{S}^{'}_i(P_{I})+\mv{S}_E\right)\mv{G}_j\right] \geq E_j, \forall j \in \mathcal{K_E} \nonumber \\
& ~~ \text{Tr}\left(\mv{S}_E\right) \leq P_{\text{sum}} - P_{I} \nonumber \\
& ~~ \mv{S}_E \succeq \mv{0}. \label{eq:ps energy}
\end{align}
Since problem (\ref{eq:ps energy}) is a standard SDP, it thus can be solved by standard convex optimization techniques, e.g. the ellipsoid method \cite{Boyd2}.

Note that problem (\ref{eq:ps energy}) can be infeasible, which means that the corresponding power allocation is not admissible. In order to find a feasible optimal power allocation between the information and energy signals, under which the WSR of all ID receivers is maximized and the harvested power constraints of all EH receivers are satisfied, we further employ bisection method to update $P_I$, as summarized in Table \ref{table3}. The convergence of this algorithm is guaranteed if problem (P1) is feasible. It is because problem (\ref{eq:ps energy}) is equivalent to the feasibility problem (\ref{eq:feasibility}) with $P_{I} = 0$.

\begin{table}[ht]
\begin{center}
\caption{\textbf{Algorithm 2}: ID Oriented separate information and energy signal design} \vspace{0.2cm}
 \hrule
\vspace{0.3cm}
\begin{enumerate}
\item {\bf Given} $P_{\text{min}}(\triangleq 0) \leq P^{*}_I < P_{\text{max}} (\triangleq P_{\text{sum}}$).
\item {\bf Repeat}
    \begin{itemize}
        \item[ a)] $P_I = \frac{1}{2}\left(P_{\text{min}} + P_{\text{max}}\right)$.
        \item[ b)] Obtain $\left\{\mv{S}^{'}_i(P_{I})\right\}$ by solving problem (\ref{eq:ps information}).
        \item[ c)] Solve problem (\ref{eq:ps energy}).
        \item[ d)] If problem (\ref{eq:ps energy}) is feasible given $P_I$, set $P_{\text{min}} \leftarrow P_I$; otherwise, set $P_{\text{max}} \leftarrow P_I$.
    \end{itemize}
\item {\bf Until} $|P_{\text{max}} - P_{\text{min}}| < \delta$ where $\delta$ is a small positive constant that controls the algorithm accuracy.
\end{enumerate}
\vspace{0.2cm} \hrule \label{table3} \end{center}
\end{table}

\subsection{EH Oriented Separate Information and Energy Signal Design}\label{sec:sub2}
In this algorithm, the energy signals are first designed to meet all the harvested power requirements at EH receivers while using the minimum transmit power, by solving the following problem:
\begin{align}
 \mathop{\mathtt{Min.}}\limits_{\mv{S}_E} &
~~ \text{Tr}\left(\mv{S}_E\right) \nonumber \\
\mathtt{s.t.}
& ~~ \text{Tr}\left[\mv{S}_E\mv{G}_j\right] \geq E_j, \forall j \in \mathcal{K_E} \nonumber \\
& ~~ \mv{S}_E \succeq \mv{0}. \label{eq:eho energy}
\end{align}
Note that (\ref{eq:eho energy}) is again a standard SDP, which thus can be solved by standard convex optimization techniques, e.g. the ellipsoid method \cite{Boyd2}.

Let $\mv{S}^{'}_E$ be the optimal solution of problem (\ref{eq:eho energy}). Then, the remaining power is allocated to information signals to maximize the WSR, i.e.,
\begin{align}
\mathop{\mathtt{Max.}}\limits_{\left\{\mv{S}_i\right\},\mv{r}} &
~~ \sum^{K_I}_{i=1}\alpha_ir_i  \nonumber \\
\mathtt{s.t.}
& ~~ \mv{r} \in \mathcal{C}_{\text{BC}}\left(\{\mv{h}_i\},\{\mv{S}_i\}\right) \nonumber \\
& ~~ \text{Tr}\left(\sum_{i \in \mathcal{K_I}}\mv{S}_i\right) \leq P_{\text{sum}} - \text{Tr}\left(\mv{S}^{'}_E\right) \nonumber \\
& ~~ \mv{S}_i \succeq \mv{0}, \forall i \in \mathcal{K_I}. \label{eq:eho information}
\end{align}
Similarly to (\ref{eq:ps information}), problem (\ref{eq:eho information}) is a WSRMax problem in MISO-BC under a single MaxLTCC, which thus can be solved. Compared with IDSIED in Section \ref{sec:sub1}, EHSIED has even lower complexity with no iterative updating of power allocation required. However, the performance of EHSIED is in general worse than IDSIED in terms of WSRMax. This is because the contribution of information signals to the EH receivers is not considered in EHSIED, such that the transmit power for energy signals is in general over-allocated.

\begin{proposition}\label{proposition:1}
If the channel of each ID receiver is orthogonal to that of any EH receiver, i.e., $\mv{h}^{H}_i\mv{g}_j = 0, \forall i \in \mathcal{K_I},\forall j \in \mathcal{K_E}$, both IDSIED and EHSIED have the same performance as Algorithm 1, i.e., IDSIED and EHSIED are both optimal.
\end{proposition}
\begin{proof}
Proposition \ref{proposition:1} can be proved by identifying the fact that with the channels of all the ID receivers being orthogonal to those of all the EH receivers, problem (P1) can be effectively decomposed into two subproblems: one for information and the other for the energy transmission design, which interact through power allocation only. Given the objective of WSRMax, it is not difficult to see that allocating the minimum power to energy transfer is optimal, i.e., EHSIED, which means the iterative power updating or IDSIED is not necessary. The details are omitted for brevity.
\end{proof}

\section{Numerical Results}\label{sec:numerical}
In this section, numerical examples are provided to validate our results. It is assumed that the signal attenuation from the AP to all EH receivers is $30$ dB corresponding to an equal distance of 5 meter, and that to all ID receivers is 70 dB at an equal distance of 20 meters. For the purpose of exposition, we define the (channel) correlation between ID receiver $i$ and EH receiver $j$ as $\rho_{i,j} = \frac{\left|\mv{h}^{H}_i\mv{g}_j\right|}{\|\mv{h}_i\|\|\mv{g}_j\|}, \forall i \in \mathcal{K_I},\forall j \in \mathcal{K_E}$. Let the correlation matrix $\boldsymbol\rho$ be the collection of all correlation coefficients with $[\boldsymbol\rho]_{i,j} = \rho_{i,j}, \forall i \in \mathcal{K_I},\forall j \in \mathcal{K_E}$. We also set the harvested power constraints of all the EH receivers identical for simplicity, i.e., $E_j = E, \forall j \in \mathcal{K_E}$. For convenience, we further denote $E_{\text{max}}$ as the maximum allowable value of $E$ for (P1) to be feasible. Note that the value of $E_{\text{max}}$ depends on the exact channel realization of all EH receivers, which can be obtained by solving the following SDP:
\begin{align}
 \mathop{\mathtt{Max.}}\limits_{E_{\text{max}},\mv{S}_E} &
~~ E_{\text{max}} \nonumber \\
\mathtt{s.t.}
& ~~ \text{Tr}\left[\mv{S}_E\mv{G}_j\right] \geq E_{\text{max}}, \forall j \in \mathcal{K_E} \nonumber \\
& ~~ \text{Tr}\left(\mv{S}_E\right) \leq P_{\text{sum}},~ \mv{S}_E \succeq \mv{0}.
\end{align}
Finally, we set $P_{\text{sum}} = 5$ Watt(W) and $\sigma^2 = -50$ dBm.

\subsection{Illustration of Optimal Information and Energy Signals}
In this subsection, we provide one numerical example to demonstrate the necessity of expanding the obtained solution to (P2) according to (\ref{eq:information}), for the case that the optimal solution of problem (P2) in (\ref{eq:solution for dual}) is not feasible to (P1). It is assumed that $N = 10$, $K_I = 1$, $K_E = 10$, $E =0.9E_{\text{max}}$, and the correlations between ID and EH receivers are distributed as $\rho_{1,j} = 2^{-j}, \forall j \in \mathcal{K_E}$, for the purpose of demonstration. The results are summarized in Table \ref{table4} (in mW), in which the harvested power of each EH receiver with the optimal solution to (P2), the optimal information signal to (P1) after expansion, and both the optimal information and energy signals to (P1) are listed in the second, third and fourth column, respectively.

By comparing the second and last columns of Table \ref{table4}, it is first observed that the obtained solution to (P2) results in imbalanced harvested power distribution among EH receivers, and in particular does not meet the minimum harvested power constraints for EH receivers $9$ and $10$. For the imbalanced distribution, it is also interesting to observe that for any two EH receivers $i, j \in \mathcal{K_E}$, $\rho_{1,i} > \rho_{1,j}$, i.e., the channel of EH receiver $i$ has higher spatial correlation with that of ID receiver, does not mean that EH receiver $i$ can harvest more power from the information embedded signal. This is because the information transfer needs to be compromised for energy transfer and shift its transmission direction away from that of rate maximization. According to the third column of Table \ref{table4}, the harvested power levels of different EH receivers are re-allocated to achieve a more balanced distribution after expanding the information signal based on (\ref{eq:information}), which confirms the results in Remark \ref{remark:3}. Finally, since $\rho_{1,10} = 2^{-10} \approx 0$, i.e, EH receiver $10$ is almost orthogonal to the ID receiver, extra dedicated energy signal is necessary to satisfy its harvested power requirement.

\begin{table}[ht]
\caption{Results on Finding primal feasible solution for problem (P1)}
\begin{center}
\begin{tabular}{|c|c|c|c|c|c|c|}
  \hline
  EH & $\text{Tr}\left[\mv{U}^{*}_1\mv{B}^{*}_1\left(\mv{U}^{*}_1\right)^{H}\mv{G}_j\right]$ & $\text{Tr}\left[\mv{S}^{*}_1\mv{G}_j\right]$ & $\text{Tr}\left[\left(\mv{S}^{*}_1 + \mv{S}^{*}_E\right)\mv{G}_j\right]$ & $E_j$\\
  receiver $j$ & (mW) & (mW) & (mW) & (mW) \\ \hline
  $1$ &$0.5035$ & $0.4995$  & $0.4995$ & $0.4995$  \\ \hline
  $2$ & $0.4995$ & $0.4995$  & $0.4995$ & $0.4995$ \\ \hline
  $3$ & $0.5013$ & $0.4995$  & $0.4995$ & $0.4995$ \\ \hline
  $4$ & $0.5005$ & $0.4995$  & $0.4995$ & $0.4995$ \\ \hline
  $5$ & $0.5079$ & $0.4995$  & $0.4995$ & $0.4995$  \\ \hline
  $6$ & $0.5000$ & $0.4995$  & $0.4995$ & $0.4995$ \\ \hline
  $7$ & $0.4996$ & $0.4995$  & $0.4995$ & $0.4995$  \\ \hline
  $8$ & $0.5057$ & $0.4995$  & $0.4995$ & $0.4995$ \\ \hline
  $9$ & $0.4881$ & $0.4995$  & $0.4995$ & $0.4995$  \\ \hline
  $10$ & $0.4603$ & $0.4693$  & $0.4995$ & $0.4995$  \\
  \hline
\end{tabular}
\end{center}\label{table4}
\end{table}

\subsection{Capacity Region Comparison}
In this subsection, we illustrate the capacity regions with and without harvested power constraints for the case of $N = 5$, $K_I = 2$ and $K_E = 3$ in Fig. \ref{fig:capacity5} and Fig. \ref{fig:capacity9}. The achievable rate regions obtained by the two benchmark algorithms, i.e., IDSIED and EHSIED, are also presented for comparison. The harvested power requirement is set to be $E = 0.5E_{\text{max}}$ for Fig. \ref{fig:capacity5first} and Fig. \ref{fig:capacity9first}, and $E = 0.9E_{\text{max}}$ for Fig. \ref{fig:capacity5second} and Fig. \ref{fig:capacity9second}. For the correlations between EH and ID receivers, we consider the following two configurations:
\begin{enumerate}
\item Highly correlated setup (HCS): for Fig. \ref{fig:capacity5}, each ID receiver is assumed to be highly correlated with all EH receivers as
\begin{align}
\boldsymbol\rho^{\text{HCS}} = \begin{bmatrix}
  1/2 & 1/4 & 1/8 \\
  1 & 1/2 & 1/4 \\
\end{bmatrix}
\end{align}
\item Less correlated setup (LCS): for Fig. \ref{fig:capacity9}, we consider a less correlated setup. In particular, it is assumed that ID receiver $1$ is orthogonal to EH receiver $2$ and ID receiver $2$ is orthogonal to all EH receivers. Thus, the correlation matrix is given as
\begin{align}
\boldsymbol\rho^{\text{LCS}} = \begin{bmatrix}
  1/2 & 0 & 1/8 \\
  0 & 0 & 0 \\
\end{bmatrix}
\end{align}
\end{enumerate}

From Fig. \ref{fig:capacity5}, it is first observed that the capacity loss with harvested power constraints for EH receivers is not significant under HCS for both the cases of $E = 0.5E_{\text{max}}$ and $0.9E_{\text{max}}$. This observation can be explained as follows: with the channels of ID receivers being highly correlated to those of EH receivers, each EH receiver can harvest significant amount of power from the information signals intending for ID receivers. As a result, the harvested power requirements are more easily satisfied while maximizing the transmission rate of ID receivers. Moreover, it is observed that EHSIED performs much worse than IDSIED and the optimal algorithm. It is because that EHSIED ignores the fact that information signals can also contribute to EH due to their broadcast property, such that only a small porion of power is allocated for information transfer. Finally, as $E$ increases, the performance gap between IDSIED and the optimal algorithm increases due to the separation of information and energy signal design.

\begin{figure}[ht]
\centering
\subfigure[]{
\includegraphics[width=0.48\textwidth]{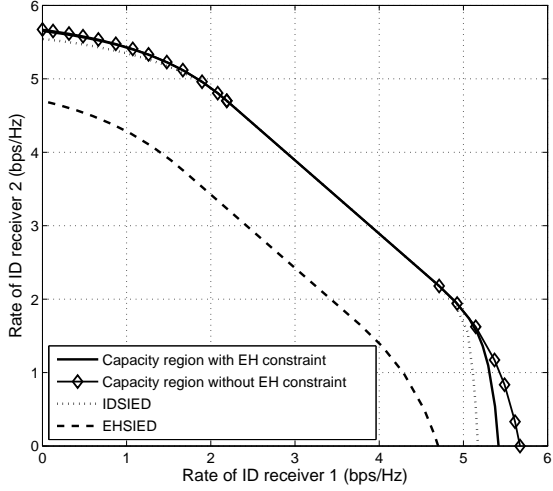}
\label{fig:capacity5first}
}
\subfigure[]{
\includegraphics[width=0.48\textwidth]{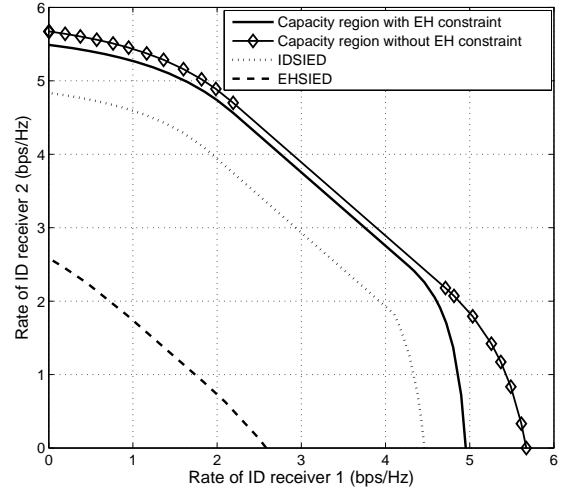}
\label{fig:capacity5second}
}
\caption[]{Capacity region under HCS: (a) $E = 0.5E_{\text{max}}$; (b) $E = 0.9E_{\text{max}}$.}
\label{fig:capacity5}
\end{figure}

From Fig. \ref{fig:capacity9}, it is observed that under LCS the capacity loss due to harvested power constraints is much larger than that under HCS (cf. Fig. \ref{fig:capacity5}), which also increases dramatically as $E$ increases. This is because the information signals for ID receivers have limited contribution to the EH receivers. One interesting result shown in Fig. \ref{fig:capacity9} is that the performance gap between the optimal and two benchmark algorithms reduces as ID receiver $2$ being given higher priority, and converges to zero while maximizing the rate of ID receiver $2$ exclusively. Since the channel of ID receiver $2$ is orthogonal to all EH receivers, problem (P1) with $\alpha_1 = 0$ and $\alpha_2 > 0$ can be decomposed into two subproblems as explained in Proposition \ref{proposition:1}. Consequently, IDSIED and EHSIED have the same performance as the optimal algorithm.

\begin{figure}[ht]
\centering
\subfigure[]{
\includegraphics[width=0.48\textwidth]{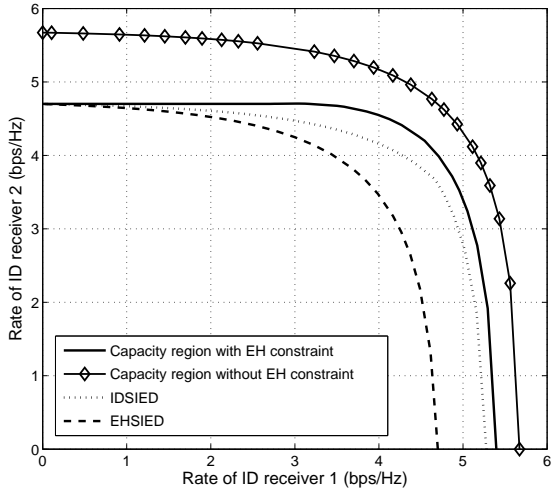}
\label{fig:capacity9first}
}
\subfigure[]{
\includegraphics[width=0.48\textwidth]{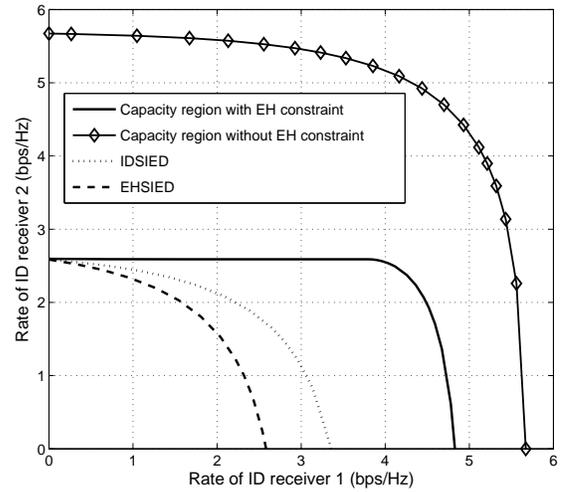}
\label{fig:capacity9second}
}
\caption[]{Capacity region under LCS: (a) $E = 0.5E_{\text{max}}$; (b) $E = 0.9E_{\text{max}}$.}
\label{fig:capacity9}
\end{figure}

\subsection{Sum-rate Comparison}
To further evaluate the performance of the optimal and two benchmark algorithms, Fig. \ref{fig:comparison1} compares their achieved sum-rate versus different EH constraint values of $E$, where the configurations for Fig. \ref{fig:sumratedemofirst} and Fig. \ref{fig:sumratedemosecond} are the same as those for Fig. \ref{fig:capacity5} and Fig. \ref{fig:capacity9}, respectively. It is first observed that the optimal algorithm outperforms both the two suboptimal algorithms, and the performance gap increases as $E$ increases. This observation further validate our theoretical results and the effectiveness of joint information and energy signals design. Note that all the three algorithms achieve the same sum-rate when $E = 0$, which is the maximum sum-rate achievable without harvested energy constraint. Second, we observe that the optimal algorithm and the IDSIED have similar performance when $E$ is small. This is because that when $E$ is sufficiently small, the information signals obtained by maximizing the sum-rate are sufficient to guarantee the harvested power constraints at each EH receiver. However, as $E$ increases, the information transfer needs to be compromised for energy transfer, such that the optimal directions of the information signals are shifted from those obtained by maximizing the sum-rate. Finally, by comparing IDSIED and EHSIED, it is observed that IDSIED outperforms EHSIED over the entire range of values of $E$. As $E$ increases, IDSIED diverges from EHSIED under HCS in Fig. \ref{fig:sumratedemofirst} but converges to EHSIED under LCS in Fig. \ref{fig:sumratedemosecond}. This is because under LCS, the information embedded signals can no longer make significant contribution to EH receivers, such that IDSIED has less noticeable advantage over EHSIED, especially when the harvested power constraints become stringent.

\begin{figure}[ht]
\centering
\subfigure[]{
\includegraphics[width=0.48\textwidth]{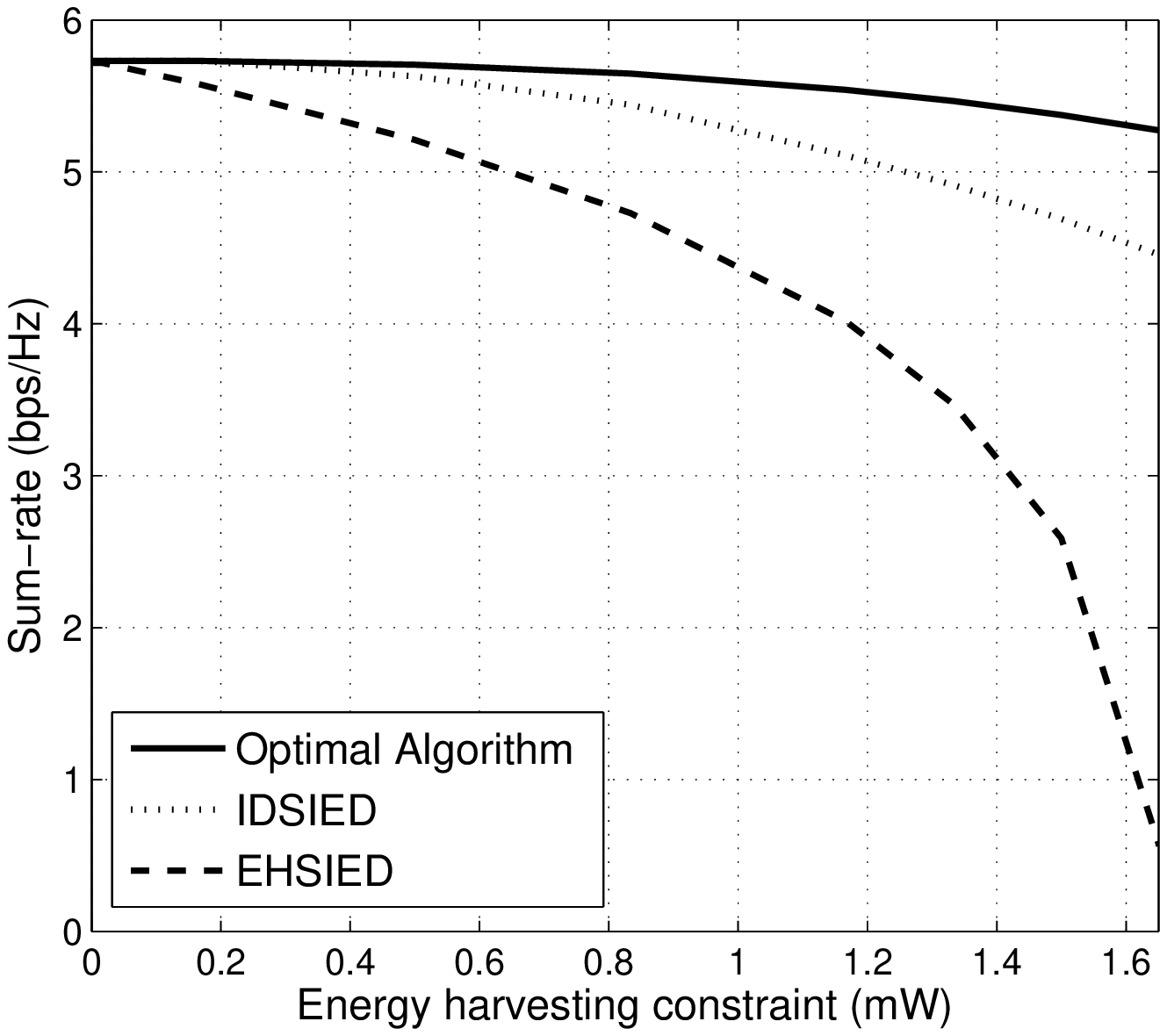}
\label{fig:sumratedemofirst}
}
\subfigure[]{
\includegraphics[width=0.48\textwidth]{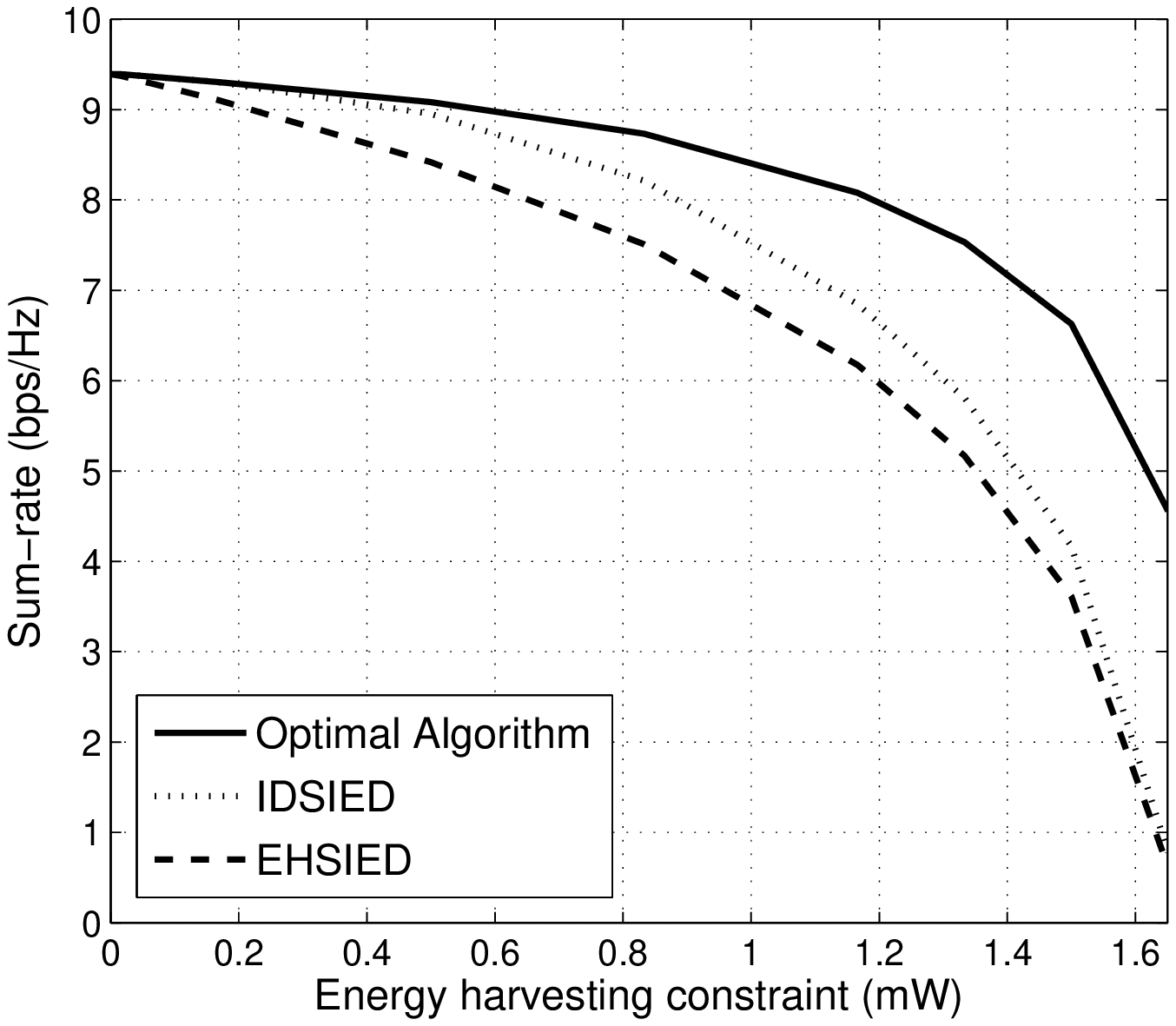}
\label{fig:sumratedemosecond}
}
\caption[]{Sum-rate performance comparison of optimal versus benchmark algorithms: (a) HCS; (b) LCS.}
\label{fig:comparison1}
\end{figure}

At last, in Fig. \ref{fig:avesumrate}, we illustrate the average sum-rate performance of the optimal and two benchmark algorithms versus different values of $E$ over $200$ randomly generated channels (various channel correlation between ID and EH receivers) for the case of $N = 5$, $K_I = 2$ and $K_E = 3$. The channel vector $\mv{h}_i$'s are generated from i.i.d. Rayleigh fading. However, due to the short transmission distance of EH receivers, for which the line-of-sight (LOS) signal is dominant, $\mv{g}_j$'s are generated based on the Rician fading model used in \cite{BF}. It is observed that, on average, the performance gap between the optimal and two benchmark algorithms increases as $E$ increases. However, the difference between IDSIED and EHSIED stays roughly the same from moderate to large values of $E$.

\begin{figure}
\centering
\epsfxsize=0.7\linewidth
\includegraphics[width=11cm]{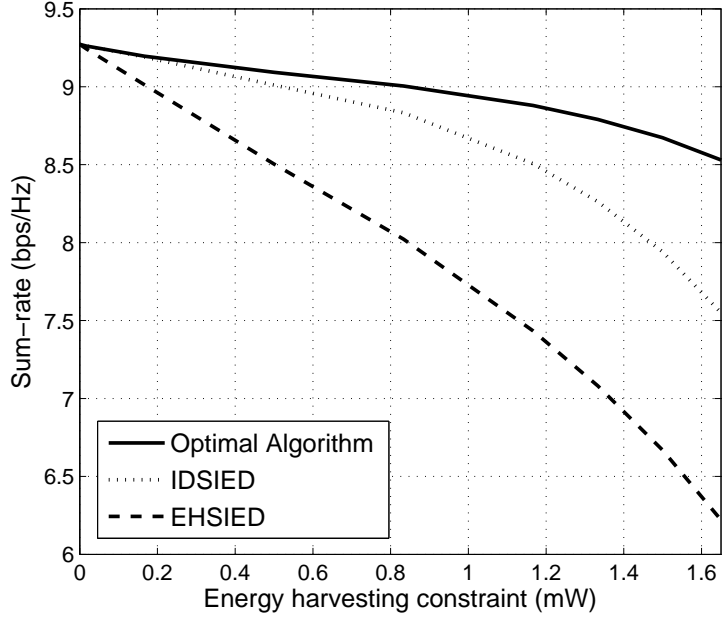}
\caption{Average sum-rate performance comparison of optimal versus benchmark algorithms.}
\label{fig:avesumrate}
\end{figure}

\section{Conclusion}\label{sec:conclusion}
In this paper, we study a MISO-BC for SWIPT, where a multi-antenna AP delivers information and energy simultaneously to multiple single-antenna receivers. We characterize the capacity region for the ID receivers by maximizing their WSR subject to the sum-power constraint at the AP and a set of minimum harvested power constraints at EH receivers. This problem corresponds to a new type of WSRMax problem for MISO-BC with combined MaxLTCC and MinLTCCs, for which a new optimal algorithm is proposed by extending the BC-MAC duality and applying the ellipsoid method. Suboptimal algorithms with separate information and energy signal designs are also presented. The proposed algorithms provide useful insights on solving general WSRMax problems with both MaxLTCCs and MinLTCCs, and the established capacity region provides a performance upper bound on all practically implementable precoding/beamforming algorithms for SWIPT in MISO-BC.

\appendices
\section{Proof of Lemma \ref{lemma:1}}\label{appendix:proof lemma 1}
The first two conditions of Lemma \ref{lemma:1} can be proved by contradiction. For convenience, we define $\mv{S}_I \triangleq \sum_{i \in \mathcal{K_I}}\mv{S}_i$ as the sum of all information covariance matrices. Furthermore, it is sufficient to consider only the case that $\mv{S}_I$ can be expressed as
\begin{align}\label{eq:lemma 1 1}
\mv{S}_I = \sum^{N}_{n=1}\mu_n\mv{u}_n\mv{u}^{H}_n
\end{align}
where $\mv{u}_n \in \mathbb{C}^{N \times 1}$ is the $n$th eigenvector of $\mv{A}$, i.e., $[\mv{u}_1,\cdots,\mv{u}_N] = [\mv{U}_1,\mv{U}_2]$ from (\ref{eq:evd}), and $\mu_n$ is a non-negative real number, $n=1,\cdots,N$. As a result, $\sum_{i\in \mathcal{K_I}}\text{Tr}(\mv{A}\mv{S}_i) = \text{Tr}\left(\mv{A}\mv{S}_I\right)$ can be expressed as $\sum^{N}_{n=1}\mu_n\mv{u}^{H}_n\mv{A}\mv{u}_n$.

Suppose that $\mv{A} \nsucceq \mv{0}$, i.e., at least one of the eigenvalues of $\mv{A}$ is negative, and $g(\{\lambda_j\})$ has an upper bounded value or $g(\{\lambda_j\}) < +\infty$. Without loss of generality, we assume that $\mv{u}_k$ is one of the eigenvectors associated with the negative eigenvalues of $\mv{A}$. Then, it follows that $\mv{u}^{H}_k\mv{A}\mv{u}_k< 0$ and $\mu_k\mv{u}^{H}_k\mv{A}\mv{u}_k \rightarrow -\infty$ as $\mu_k$ approaches $+\infty$. Therefore, it is easy to verify that by choosing $\mv{S}_I$ based on (\ref{eq:lemma 1 1}) and $\mv{S}_i = \frac{1}{K_I}\mv{S}_I, \forall i \in \mathcal{K_I}$ with $\mu_k$ being large enough, $\mu_i, \forall i \neq k$ can be set to be arbitrary large such that we can achieve arbitrary large WSR for ID receivers without violating (\ref{eq:modified power}), which results in $g(\{\lambda_j\}) = +\infty$. Consequently, $\mv{A}$ has to be positive semi-definite. Since similar arguments can be used to verify the second condition of Lemma \ref{lemma:1}, the details are omitted for brevity.

Next, we prove the third condition of Lemma \ref{lemma:1}. Given $\mv{A}$ being positive semi-definite, it has a positive semi-definite square root, i.e., $\mv{A} = \mv{A}^{1/2}\mv{A}^{1/2}$. Therefore, $\text{Tr}(\mv{A}\mv{S}_I)$ and $\text{Tr}(\mv{A}\mv{S}_E)$ can be expressed as $\text{Tr}(\mv{A}^{1/2}\mv{S}_I\mv{A}^{1/2})$ and $\text{Tr}(\mv{A}^{1/2}\mv{S}_E\mv{A}^{1/2})$, respectively. Since both $\mv{A}^{1/2}\mv{S}_I\mv{A}^{1/2}$ and $\mv{A}^{1/2}\mv{S}_E\mv{A}^{1/2}$ are positive semi-definite, it follows that $\text{Tr}(\mv{A}^{1/2}\mv{S}_I\mv{A}^{1/2}) \geq 0$ and $\text{Tr}(\mv{A}^{1/2}\mv{S}_E\mv{A}^{1/2}) \geq 0$. Lemma \ref{lemma:1} is thus proved.

\section{Proof of Lemma \ref{lemma:2}}\label{appendix:proof lemma 2}
From the proof of Lemma \ref{lemma:1} in Appendix \ref{appendix:proof lemma 1}, $\sum_{i\in \mathcal{K_I}}\text{Tr}(\mv{A}\mv{S}_i) \geq 0$ and $\text{Tr}(\mv{A}\mv{S}_E) \geq 0$. Given the fact that only $\{\mv{S}_i\}$ is related to the information transfer, any solution to problem (\ref{eq:combined problem}) with $\text{Tr}(\mv{A}\mv{S}_E) > 0$ reduces the transmit power allocated to the information transfer and is thus suboptimal. Therefore, the optimal energy covariance matrix needs to satisfy $\text{Tr}(\mv{A}\mv{\bar{S}}_E )= 0$ equivalently $\mv{A}\mv{\bar{S}}_E = \mv{0}$, which means $\mv{\bar{S}}_E$ lies in the null space of $\mv{A}$. According to (\ref{eq:evd}), the vectors in $\mv{U}_2$ form the orthogonal basis for the null space of $\mv{A}$. Therefore, $\mv{\bar{S}}_E$ in general can be expressed as $\mv{\bar{S}}_E = \mv{U}_2\mv{\bar{E}}\mv{U}_2^{H}$, where $\mv{\bar{E}} \in \mathbb{C}^{(N-m)\times (N-m)}$ is any positive semi-definite matrix. Note that for case of $m = N$, i.e., $\mv{A}$ is of full rank, $\mv{A}\mv{\bar{S}}_E = \mv{0}$ implies that $\mv{\bar{S}}_E = \mv{0}$. Lemma \ref{lemma:2} is thus proved.

\section{Proof of Lemma \ref{lemma:4}}\label{appendix:proof lemma 4}
Without loss of generality, $\mv{S}_i$ can be expressed as
\begin{align}
\mv{S}_i & = \left[\mv{U}_1, \mv{U}_2\right]\left[\begin{array}{ll}
\mv{B}_i  &\mv{C}_i \\
\mv{C}^{H}_i  &\mv{D}_i
\end{array}\right]\left[\mv{U}_1, \mv{U}_2\right]^{H} \\
& = \mv{U}_1\mv{B}_i\mv{U}^{H}_1 + \mv{U}_1\mv{C}_i\mv{U}^{H}_2 + \mv{U}_2\mv{C}^{H}_i\mv{U}^{H}_1 + \mv{U}_2\mv{D}_i\mv{U}^{H}_2 \label{eq:decomposition}
\end{align}
where $\mv{B}_i \in \mathbb{C}^{m\times m}$, $\mv{D}_i \in \mathbb{C}^{(N-m)\times (N-m)}$ and  $\mv{C}_i \in \mathbb{C}^{m\times (N-m)}$, $\forall i \in \mathcal{K_I}$.  Note that $\mv{B}_i = \mv{B}^{H}_i$ and $\mv{D}_i = \mv{D}^{H}_i$. Since $\mv{U}_2$ lies in the null space of $\mv{A}$ (from Lemma \ref{lemma:2}) and consequently in the null space of $\mv{H}$ (from Lemma \ref{lemma:1}), it is observed that $r_i$ and $\sum_{i\in\mathcal{K_I}}\text{Tr}(\mv{A}\mv{S}_i)$ do not depend on $\mv{C}_i$ and $\mv{D}_i, \forall i \in \mathcal{K_I}$. Consequently, it is optimal to set $\mv{C}_i=\mv{0}$ and $\mv{D}_i = \mv{0}, \forall i \in \mathcal{K_I}$, and accordingly problem (\ref{eq:no energy}) with given $\{\lambda_{j}\}$ can be further simplified as (\ref{eq:full rank}) given in Lemma \ref{lemma:4}. With $\mv{\hat{A}}$ being full rank, problem (\ref{eq:full rank}) can be solved by the general BC-MAC duality as in \cite{UDD}, and results in unique rank-one information covariance matrices, i.e., $\mv{U}_1\mv{\bar{B}}_i\mv{U}^{H}_1, i \in \mathcal{K_I}$, the details of which are illustrated below.

Without loss of generality, we assume that $\alpha_1 \geq \alpha_2 \geq \cdots \geq \alpha_{K_I} \geq 0$. For the MISO-BC, its dual single-input multiple-output (SIMO) MAC consists of $K_I$ single-antenna transmitters that send independent information to one common receiver with $N$ antennas. At transmitter $i, i \in \mathcal{K_I}$, let $p_i$ be its transmit power, $s^{(m)}_i$ be a CSCG random variable representing its transmitted information signal, and $\mv{\hat{h}}^{H}_i$ be its channel vector to the receiver. Then the received signal in the dual SIMO-MAC is expressed as
\begin{align}
\mv{y}^{(m)} = \sum^{K_I}_{i=1}\mv{\hat{h}}^{H}_i\sqrt{p_i}s^{(m)}_i + \mv{z}^{(m)}
\end{align}
where $\mv{z}^{(m)}\thicksim\mathcal{CN}\left(\mathbf{0},\mv{\hat{A}}\right)$.

According to \cite{UDD}, problem (\ref{eq:full rank}) is equivalent to its dual MAC problem expressed as
\begin{align}
\mathop{\mathtt{Max.}}\limits_{\{p_i \geq 0\}} &
~~ \sum^{K_I}_{i=1}\alpha_ir^{(m)}_i  \nonumber \\
\mathtt{s.t.}
& ~~ \sum^{K_I}_{i=1}p_i \leq P_{A} \label{eq:mac}
\end{align}
where $r^{(m)}_i$ is given as
\begin{align}
\log_2\frac{\left|\mv{\hat{A}}+\sum^{i}_{k=1}p_k\mv{\hat{h}}_k\mv{\hat{h}}^{H}_k\right|}{\left|\mv{\hat{A}}+\sum^{i-1}_{k=1}p_k\mv{\hat{h}}_k\mv{\hat{h}}^{H}_k\right|}
\end{align}
due to the polymatroid structure of the MAC capacity region \cite{Mac}, and the user decoding order being determined by the magnitude of $\alpha_i$'s. Since problem (\ref{eq:mac}) is convex, it can be solved efficiently via standard convex optimization techniques. With the optimal solution to problem (\ref{eq:mac}), i.e., $\{p^{\star}_i\}$, at hand, the optimal receive beamforming vector can be obtained based on the minimum-mean-squared-error (MMSE) principle as
\begin{align}
\mv{v}^{*}_i = \frac{\left(\mv{\hat{A}}+\sum^{i-1}_{k=1}p^{*}_k\mv{\hat{h}}_k\mv{\hat{h}}^{H}_k\right)^{-1}\mv{\hat{h}}_i}{\left\|\left(\mv{\hat{A}}+\sum^{i-1}_{k=1}p^{*}_k\mv{\hat{h}}_k\mv{\hat{h}}^{H}_k\right)^{-1}\mv{\hat{h}}_i\right\|}, \forall i \in \mathcal{K_I}.
\end{align}

After obtaining the optimal solution of $\{\mv{v}^{*}_i, p^{*}_i\}$ for the uplink problem (\ref{eq:mac}), we then map the solution to $\left\{\mv{w}^{*}_i\right\}$ for the downlink problem (\ref{eq:full rank}). As shown in \cite{UDD}, since the downlink transmit beamforming vectors are identical to the uplink receive beamforming vectors up to certain scaling factors, $\mv{w}^{*}_i$ can be expressed as $\mv{w}^{*}_i = \sqrt{q^{*}_i}\mv{v}^{*}_i, \forall i \in \mathcal{K_I}$. Furthermore, the rate-tuples achieved for both the BC and MAC are identical. Therefore, the following set of equations can be utilized to find $\{q^{*}_i\}$:
\begin{align}
\log_2\left(1 + \frac{q^{*}_i|\mv{\hat{h}}^{H}_i\mv{v}^{*}_i|^2}{\sum^{K_I}_{k=i+1}q^{*}_k|\mv{\hat{h}}^{H}_i\mv{v}^{*}_k|^2+1}\right) = \log_2\frac{\left|\mv{\hat{A}}+\sum^{i}_{k=1}p^{*}_k\mv{\hat{h}}_k\mv{\hat{h}}^{H}_k\right|}{\left|\mv{\hat{A}}+\sum^{i-1}_{k=1}p^{*}_k\mv{\hat{h}}_k\mv{\hat{h}}^{H}_k\right|}, \forall i \in \mathcal{K_I}
\end{align}
i.e.,
\begin{align}
q^{*}_i = \frac{2^{\left(r^{(m)}_i\right)^{*}}-1}{|\mv{\hat{h}}^{H}_i\mv{v}^{*}_i|^2}\left(\sum^{K_I}_{k=i+1}q^{*}_k|\mv{\hat{h}}^{H}_i\mv{v}^{*}_k|^2+1\right), \forall i \in \mathcal{K_I}.
\end{align}
Finally, the optimal solution to problem (\ref{eq:full rank}) can be computed as
\begin{align}
\mv{\bar{B}}_i = \mv{w}^{*}_i(\mv{w}^{*}_i)^{H}, \forall i \in \mathcal{K_I}.
\end{align}
Lemma \ref{lemma:4} is thus proved.

\section{Proof of Lemma \ref{lemma:3}}\label{appendix:proof lemma 3}
Since the encoding order of the BC is the reverse of the decoding order of its dual MAC \cite{UDD}, which can be obtained from Section \ref{sec:BC MAC} while solving problem (\ref{eq:full rank}) and is assumed to be in accordance with the ID receiver index without loss of generality, problem (\ref{eq:no energy}) can now be written explicitly as
\begin{align}
\mathop{\mathtt{Max.}}\limits_{\left\{\mv{S}_i\right\},\mv{r}} &
~~ \sum^{K_I}_{i=1}\alpha_ir_i  \nonumber \\
\mathtt{s.t.}
& ~~ \sum_{i \in \mathcal{K_I}}\text{Tr}(\mv{A}\mv{S}_i) \leq P_{A} \nonumber \\
& ~~ \mv{S}_i \succeq \mv{0}, \forall i \in \mathcal{K_I} \label{eq:BC with order}
\end{align}
where $r_i$ is given by
\begin{align}\label{eq:rate}
r_i = \log_2\left(\frac{\sigma^2+\mv{h}^{H}_i\left(\sum^{K_I}_{k=i}\mv{S}_k\right)\mv{h}_i}{\sigma^2+\mv{h}^{H}_i\left(\sum^{K_I}_{k=i+1}\mv{S}_k\right)\mv{h}_i}\right).
\end{align}

The KKT optimality conditions of problem (\ref{eq:BC with order}) are given by
\begin{align}
&\frac{\partial\sum^{K_I}_{i=1}r_i}{\partial\mv{S}_i}  = \omega\mv{A} + \mv{\Psi}_i, \forall i \in \mathcal{K_I} \nonumber \\
&\omega\left[\sum_{i \in \mathcal{K_I}}\text{Tr}(\mv{A}\mv{S}_i)-P_{A}\right] = 0 \nonumber \\
&\text{Tr}\left(\mv{\Psi}_i\mv{S}_i\right) = 0, \forall i \in \mathcal{K_I} \label{eq:KKT P2S}
\end{align}
where $\omega \geq 0$ and $\mv{\Psi}_i \succeq \mv{0}, \forall i \in \mathcal{K_I}$ are the Lagrange multipliers associated with $\sum_{i \in \mathcal{K_I}}\text{Tr}(\mv{A}\mv{S}_i) \leq P_{A}$ and $\mv{S}_i \succeq \mv{0}, \forall i \in \mathcal{K_I}$, respectively.

This lemma can be proven by first showing that the duality gap between problem (\ref{eq:BC with order}) and its Lagrange dual problem is zero, and the KKT conditions given in (\ref{eq:KKT P2S}) are sufficient for a solution to be optimal for problem (\ref{eq:BC with order}). Since the proofs are similar that of \cite[Proposition 2]{UDD} and \cite[Proposition 3]{UDD}, they are omitted for brevity. To complete the proof, we need further show that the optimal value of problem (\ref{eq:no energy}) is equal to that of (P1) with any fixed encoding order, the details of which are given as follows.

We first consider a fixed encoding order for problem (P1) termed as problem (P1F), given by the optimal encoding order for problem (P2), which has been assumed to be the same as the ID receiver index order. Under this encoding order, the information rate for ID receiver $i$ is given in (\ref{eq:rate}).

Note that the optimal solution of problem (P1F) is a lower bound on the optimal solution of problem (P1). The KKT conditions of problem (P1F) can be written as
\begin{align}
&\frac{\partial\sum^{K_I}_{i=1}r_i}{\partial\mv{S}_i}  = \theta_0\mv{I} - \sum^{K_E}_{j=1}\theta_j\mv{G}_j + \mv{\Omega}_i, \forall i \in \mathcal{K_I} \\
&\frac{\partial\sum^{K_I}_{i=1}r_i}{\partial\mv{S}_E}  = \theta_0\mv{I} - \sum^{K_E}_{j=1}\theta_j\mv{G}_j + \mv{\Omega}_E\\
&\theta_j\left(\text{Tr}\left[(\mv{S}_I+\mv{S}_E)\mv{G}_j\right]-E_j\right) = 0, \forall j \in \mathcal{K_E} \\
&\theta_0\left(\text{Tr}\left[\mv{S}_I+\mv{S}_E\right]-P_{\text{sum}}\right) = 0
\end{align}
where $\{\theta_j\}^{K_E}_{j=1}$, $\theta_0$, $\{\mv{\Omega}_i\}$ and $\mv{\Omega}_E$ are the Lagrange multipliers with respect to the constraints in (\ref{eq:P1 1}), (\ref{eq:P1 2}) and (\ref{eq:P1 3}), respectively. For convenience, we define $\mv{S}_I \triangleq \sum_{i \in \mathcal{K_I}}\mv{S}_i$. When the optimal solution of problem (P1F) is achieved, we assume that the corresponding optimal primal and dual solutions are $\mv{\tilde{S}}_I$, $\mv{\tilde{S}}_E$, $\{\tilde{\theta}_j\}^{K_E}_{j=1}$, $\tilde{\theta}_0$, $\{\mv{\tilde{\Omega}}_i\}$ and $\mv{\tilde{\Omega}}_E$.

We now write the KKT conditions of problem (\ref{eq:no energy}) with $\lambda_0 = \tilde{\theta}_0$ and $\lambda_j = \tilde{\theta}_j, \forall j$, as follows:
\begin{align}
&\frac{\partial\sum^{K_I}_{i=1}r_i}{\partial\mv{S}_i}  = \omega\left(\tilde{\theta_0}\mv{I} - \sum^{K_E}_{j=1}\tilde{\theta_j}\mv{G}_j\right) + \mv{\Psi}_i, \forall i \in \mathcal{K_I} \nonumber \\
&\omega\left[\tilde{\theta_0}\text{Tr}\left(\mv{S}_I\right) - \sum^{K_E}_{j=1}\tilde{\theta}_j\text{Tr}\left(\mv{S}_I\mv{G}_j\right)-\tilde{\theta_0}P_{\text{sum}} + \sum^{K_E}_{j=1}\tilde{\theta}_jE_j\right] = 0. \label{eq:KKT need prove}
\end{align}
If we choose $\mv{S}_I = \mv{\tilde{S}}_I + \mv{\tilde{S}}^{\mv{\tilde{A}}}_E$, where $\mv{\tilde{S}}^{\mv{\tilde{A}}}_E = \mv{\tilde{U}}_1\mv{\tilde{U}}^{H}_1\mv{\tilde{S}}_E\mv{\tilde{U}}_1\mv{\tilde{U}}^{H}_1$ and $\mv{\tilde{U}}_1$ consists of the orthogonal basis defining the range of $\mv{\tilde{A}} = \tilde{\theta_0}\mv{I} - \sum^{K_E}_{j=1}\tilde{\theta_j}\mv{G}_j$ similar as that in (\ref{eq:evd}), $\omega=1$, and $\mv{\Psi}_i = \mv{\tilde{\Omega}}_i, \forall i \in \mathcal{K_I}$, then KKT conditions in (\ref{eq:KKT need prove}) are satisfied. According to the fact that the duality gap between problem (\ref{eq:BC with order}) and its Lagrange dual problem is zero, $\mv{\tilde{S}}_I + \mv{\tilde{S}}^{\mv{\tilde{A}}}_E$ is optimal for problem (\ref{eq:no energy}). Therefore, the optimal value of problem (\ref{eq:no energy}) with $\lambda_0 = \tilde{\theta}_0$ and $\lambda_j = \tilde{\theta}_j, \forall j$, is equal to the optimal value of problem (P1F). Therefore, the optimal value of problem (P1F), which is a lower bound on the optimal value of problem (P1), meets the optimal value of problem (\ref{eq:no energy}) with $\lambda_0 = \tilde{\theta}_0$ and $\lambda_j = \tilde{\theta}_j, \forall j$, which is an upper bound on the optimal value of problem (P1). The above results also imply that the minimum value of $g(\{\lambda_j\})$ over $\{\lambda_j\}$ is achieved when $\lambda_0 = \tilde{\theta}_0$ and $\lambda_j = \tilde{\theta}_j, \forall j$. The proof of Lemma \ref{lemma:3} thus follows.

\section{Proof of Lemma \ref{lemma:5} and Lemma \ref{lemma:7}}\label{appendix:proof lemma 7}
We start with proving Lemma \ref{lemma:5}. It is first observed that the condition $\text{Null}(\mv{A}) \subseteq \text{Null}\left(\mv{H}\right)$ is equivalent to that $\mv{v}_i \nsubseteq \text{Null}(\mv{A}), \forall i \leq t$, where $t$ denotes the rank of matrix $\mv{H}$, and $\mv{v}_i, i=1,\cdots,t$, denote the left singular vectors of $\mv{H}$ corresponding to its non-zero singular values. Furthermore, given $\mv{A} \succeq \mathbf{0}$, the condition $\mv{v}_i \nsubseteq \text{Null}(\mv{A}), \forall i \leq t$, can be further expressed as $\mv{v}^{H}_i \mv{A}\mv{v}_i> 0, \forall i \leq t$. The proof of Lemma \ref{lemma:5} thus follows.

Next, we proceed to show Lemma \ref{lemma:7}. For the purpose of illustration, we define $\mv{F}(\mv{\lambda}) = - \lambda_0\mv{I} + \sum^{K_E}_{j=1}\lambda_j\mv{G}_j$, where $\mv{\lambda} = [\lambda_0,\cdots,\lambda_{K_E}]^T$. Then the constraint in (\ref{eq:p3 2}) is equivalent to $\mv{F}(\mv{\lambda}) \preceq \mv{0}$. First, the semi-definite constraint $\mv{F}(\mv{\lambda}) \preceq \mv{0}$ can be equivalently expressed as a scalar inequality constraint as
\begin{align}
f(\boldsymbol\lambda) \triangleq \lambda_{\text{max}}\left(\mv{F}(\boldsymbol\lambda)\right) \leq 0
\end{align}
where $\lambda_{\text{max}}\left(\cdot\right)$ denotes the largest eigenvalue. Thus, the above constraint can be equivalently written as
\begin{align}
f(\boldsymbol\lambda) = \max_{\|\mv{z}\|^2=1}\mv{z}^{H}\mv{F}(\boldsymbol\lambda)\mv{z} \leq 0.
\end{align}

Given a query point $\boldsymbol\lambda_1 = [\lambda_{0,1},\cdots,\lambda_{K_E,1}]^T$, we can find the normalized eigenvector $\mv{z}_1$ of $\mv{F}(\boldsymbol\lambda_1)$ corresponding to $\lambda_{\text{max}}\left(\mv{F}(\boldsymbol\lambda_1)\right)$. Consequently, we can determine the value of the scalar constraint at a query point as $f(\boldsymbol\lambda_1) = \mv{z}_1^{H}\mv{F}(\boldsymbol\lambda_1)\mv{z}_1 = \lambda_{\text{max}}\left(\mv{F}(\boldsymbol\lambda_1)\right)$. To obtain a subgradient, we show the following:
\begin{align}
f(\boldsymbol\lambda) - f(\boldsymbol\lambda_1) & = \max_{\|\mv{z}\|^2=1}\mv{z}^{H}\mv{F}(\boldsymbol\lambda)\mv{z} - \mv{z}_1^{H}\mv{F}(\boldsymbol\lambda_1)\mv{z}_1 \\
& \geq \mv{z}_1^{H}\left[\mv{F}(\boldsymbol\lambda_1) - \mv{F}(\boldsymbol\lambda)\right]\mv{z}_1\\
& = \|\mv{z}_1\|^2(\lambda_0 - \lambda_{0,1}) - \sum^{K_E}_{j=1}\left(\mv{z}_1^{H}\mv{G}_j\mv{z}_1\right)(\lambda_j - \lambda_{j,1})
\end{align}
where the last equality follows from the affine structure of the semi-definite constraint. Lemma \ref{lemma:7} thus follows.

\end{document}